\documentclass[DIV=classic,a4paper,10pt]{myart}

\KOMAoptions{DIV=last}

\usepackage[normalem]{ulem}
\usepackage{kpfonts}
\usepackage[T1]{fontenc}

\usepackage{caption}
\usepackage{subcaption}

\newcommand{\norm}[1]{\left\Vert#1\right\Vert}

\newcommand*{\cadlag}{c\`adl\`ag }

\DeclareMathOperator*{\esssup}{ess\,sup}

\begin{document}
 
\title{Characterization of Fully Coupled FBSDE in Terms of Portfolio Optimization}

\author[a,1,t1,t3]{Samuel Drapeau}
\author[b,2,t4]{Peng Luo}
\author[c,3,t2]{Dewen Xiong}

\address[a]{SAIF/CAFR/CMAR and School of Mathematical Sciences, Shanghai Jiao Tong University, China}
\address[b]{Department of Statistics and Actuarial Sciences, University of Waterloo, Canada}
\address[c]{School of Mathematical Sciences, Shanghai Jiao Tong University, China}
\eMail[1]{sdrapeau@saif.sjtu.edu.cn}
\eMail[2]{peng.luo@uwaterloo.ca}
\eMail[3]{xiongdewen@sjtu.edu.cn}

\myThanks[t1]{Financial support from the National Science Foundation of China, Grant number 11971310.}
\myThanks[t2]{Financial support from the National Science Foundation of China, Grant number 11671257.}
\myThanks[t3]{Financial support from Shanghai Jiao Tong University, Grant ``Assessment of Risk and Uncertainty in Finance'' number AF0710020.}
\myThanks[t4]{Financial support from the Natural Sciences and Engineering Research Council of Canada, Grant RGPIN-2017-04054.}

\abstract{
We provide a verification and characterization result of optimal maximal sub-solutions of BSDEs in terms of fully coupled forward backward stochastic differential equations.
We illustrate the application thereof in utility optimization with random endowment under probability and discounting uncertainty.
We show with explicit examples how to quantify the costs of incompleteness when using utility indifference pricing, as well as a way to find optimal solutions for recursive utilities.
}
\keyWords{Fully Coupled FBSDE, Utility Portfolio Optimization, Random Endowment, Probability and Discounting Uncertainty.}

\keyAMSClassification{60H20 - 93E20 - 91B16 - 91G10} 

\date{\today}
\maketitle

\section{Introduction}
Our motivation is the study of the classical portfolio optimization as follows:
In a Brownian filtrated probability space, we consider an agent having a random endowment -- or contingent claim -- $F$ delivering at time $T$.
Starting with an initial wealth $x$, she additionally has the opportunity to invest with a strategy $\hat{\pi}$ in a financial market with $n$ stocks $\hat{S}=(S^1,\ldots,S^n)$ resulting in a corresponding wealth process
\begin{equation*}
    X^{\hat{\pi}}_t=x+\int_{0}^{t} \hat{\pi} \cdot \frac{d\hat{S}}{\hat{S}}
\end{equation*}
where $d\hat{S}/\hat{S}:=(dS^1/S^1, \ldots, dS^n/S^n)$.
She intends to choose a strategy $\hat{\pi}^\ast$ as to optimize her utility in the sense that
\begin{equation*}
    U\left( F+X^{\hat{\pi}^\ast}_T \right)\geq U\left( F+X^{\hat{\pi}}_T \right) \quad \text{for all admissible strategies }\hat{\pi}.
\end{equation*}
Hereby, $F\mapsto U(F)$ is a general utility function -- quasi-concave and increasing -- mapping random variables to $[-\infty,\infty)$.\footnote{On the one hand, quasi-concavity reflects the underlying convexity of general preference ordering in terms of diversification, and on the other hand, monotonicity is a consequence of preferences for better outcomes, see \citep{marinacci2011,drapeau2013} for instance.}
For instance $U(Y)=u^{-1}(E[u(Y)])$ where $u:\mathbb{R}\to \mathbb{R}$ is an increasing concave function corresponding to the certainty equivalent of the classical expected utility \`a la \citet{neumann1947} and \citet{savage1972}.
It may however be a more general concave and increasing operator given by non-linear expectations -- solutions of concave backward stochastic differential equations --  introduced by \citet{peng1997}.
In this setting the utility $U(F)$ is given by the value $Y_0$, solution at time $0$ of the concave backward stochastic differential equation
\begin{equation*}
    Y_t = F-\int_{t}^{T}g(Y,Z)ds-\int_{t}^{T}Z \cdot dW
\end{equation*}
for a jointly convex Lipschitz generator $g:\mathbb{R}\times \mathbb{R}^d\to \mathbb{R}$ and $W$ is a $d$-dimensional Brownian motion.
This functional is concave and increasing.
Recently, \citet{heyne2013} introduced the concept of minimal super-solution of convex backward stochastic differential equations -- in this paper maximal sub-solutions of concave backward stochastic differential equations -- to extend the existence domain of classical backward stochastic differential equations for generator having arbitrary growth.
In this context, the utility $U(F)$ is given by the value $Y_0$, maximal sub-solution of the concave backward stochastic differential equation
\begin{equation}\label{eq:max-subsol-int01}
    \begin{cases}
        Y_s & \leq \displaystyle Y_t-\int_{s}^{t}g(Y,Z)du-\int_{s}^{t}Z \cdot dW, \quad 0\leq s\leq t\leq T\\
        \\
        Y_T & = F
    \end{cases}
\end{equation}
This functional $F\mapsto U(F)$ is also concave and increasing and therefore a utility functional.
Furthermore, according to \citet{DrapeauTangpi}, it admits a dual representation
\begin{equation*}
    U(F)=\inf_{b,c}\left( E\left[ D_T^bM_T^{c} F+\int_{0}^{T}D^bM^{c}g^\ast(b,c)ds  \right] \right)
\end{equation*}
where $g^\ast$ is the convex conjugate of the generator $g$, $D^b=\exp(-\int_{}^{} bds)$ is a discounting factor and $M^c:=\exp(-\int_{}^{} c\cdot dW-\int_{}^{} c^2/2 ds)$ is a probability density.
The interpretation of this utility functional is that it assesses probability uncertainty, as for monetary risk measures see \citep{foellmer2002}, as well as discounting uncertainty, as for sub-cash additive functional see \citep{ravanelli2009}.

Assuming $1\leq n\leq d$ and taking the utility $U$ defined as the value at $0$ of the maximal sub-solution of \eqref{eq:max-subsol-int01}, we want to find a strategy $\hat{\pi}^\ast$ maximizing $U(F+X^{\hat{\pi}}_T)$.
Given the corresponding maximal sub-solution $(Y,Z)$ of \eqref{eq:max-subsol-int01} such that $Y_0 = U(F+X^{\hat{\pi}}_T)$, proceeding to the variable change
\begin{equation*}
    \bar{Y} := Y - X^{\hat{\pi}} \quad \text{and} \quad \bar{Z} := Z - \pi   
\end{equation*}
where\footnote{For $z\in\mathbb{R}^d$, we will use the notation $z=(\hat{z},\tilde{z})$ where $\hat{z}$ and $\tilde{z}$ denote the first $d$ and the last $d-n$ components of $z$ and make the convention that $z=(\hat{z},\tilde{z})=\hat{z}$ if $n=d$. } $\pi=(\hat{\pi}, 0)$ leads to the following equivalent formulation in terms of the following forward backward stochastic system
\begin{equation}\label{eq:max-subsol}
     \begin{cases}
         X_T^{\hat{\pi}} & = \displaystyle x + \int_{0}^{t}\hat{\pi}\cdot \hat{\theta} dt + \int_{0}^{t}\hat{\pi} \cdot d\hat{W}  \\
         \bar{Y}_s & \leq \displaystyle Y_t-\int_{s}^{t}\left[g(X^{\hat{\pi}}+\bar{Y},\bar{Z}+\pi) - \hat{\pi}\cdot \hat{\theta}\right]du - \int_{s}^{t}\bar{Z} \cdot dW ,\quad 0\leq s\leq t\leq T\\
         \bar{Y}_T & = F
    \end{cases}
\end{equation}
for some bounded market price of risk $\hat{\theta}$.
Transferring the terminal dependence on the forward part to the generator allows to state the main results of this paper, namely, a verification and characterization of an optimal strategy $\hat{\pi}^\ast$ in terms of the following fully coupled forward backward stochastic differential equation
\begin{equation}\label{eq:mainsystem}
    \begin{cases}
        X_t & = \displaystyle x + \int_{0}^{t}\hat{\pi}(X+\bar{Y},\bar{Z},\hat{V})\cdot \hat{\theta} ds+ \int_{0}^{t} \hat{\pi}(X+\bar{Y},\bar{Z},\hat{V})\cdot d\hat{W}\\
        \bar{Y}_t & = \displaystyle F - \int_{t}^{T}\left[g\left( \bar{Y}+X,\bar{Z}+\pi(X+\bar{Y},\bar{Z},\hat{V})\right) -\hat{\pi}(X+\bar{Y},\bar{Z},\hat{V})\cdot \hat{\theta}\right] ds -\int_{t}^{T}\bar{Z} \cdot dW\\
        U_t & \displaystyle = U_T+\int_{t}^{T} \left(\frac{\hat{V}^2-\tilde{V}^2}{2} + \hat{V}\cdot \hat{\theta}\right) ds +\int_{t}^{T}V\cdot d W\\
        U_T &=\displaystyle \int_{0}^{T} \left(\partial_yg\left( X+\bar{Y}, \bar{Z}+\pi(X+\bar{Y},\bar{Z},\hat{V}) \right)+\frac{\partial_{\tilde{z}}g\left(X+\bar{Y},\bar{Z}+\pi(X+\bar{Y},\bar{Z},\hat{V}) \right)^2}{2}\right)ds\\
            & \displaystyle \quad\quad\quad +\int_{0}^{T}\partial_{\tilde{z}}g\left( X+\bar{Y},\bar{Z}+\pi(X+\bar{Y},\bar{Z},\hat{V})\right)\cdot d\tilde{W}
    \end{cases}
\end{equation}
where
\begin{itemize}
    \item $W=(\hat{W},\tilde{W})$ is a $d$ dimensional Brownian motion whereby $\hat{W}$ and $\tilde{W}$ denote the first $n$ and last $d-n$ components, respectively;
    \item $g$ is a convex generator;
    \item $F$ is a bounded terminal condition.
    \item $\pi(y,z,\hat{v}):=(\hat{\pi}(y,z,\hat{v}),0)$ is a point-wise solution to
        \begin{equation*}
            \partial_{\hat{z}}g(y,z+\pi(y,z,\hat{v}))=\hat{v} + \hat{\theta}
        \end{equation*}
\end{itemize}
and the optimal strategy is given by $\hat{\pi}^\ast = \hat{\pi}(X+\bar{Y},Z,\hat{V})$.

As for maximal sub-solutions of backward stochastic differential equations introduced and studied by \citet{heyne2013,heyne2014}, they can be understood as an extension of backward stochastic differential equations, where equality is dropped in favor of inequality allowing weaker conditions for the generator $g$.
It allows to achieve existence, uniqueness and comparison theorem without growth assumptions on the generator as well as weaker integrability condition on the forward process and terminal condition.
To stress the relation between maximal sub-solutions and solutions of backward stochastic differential equations, maximal sub-solutions can be characterized as maximal viscosity sub-solutions in the Markovian case, see \citep{mainberger2016}.
It also turns out that they are particularly adequate for optimization problem in terms of convexity or duality among others, see \citep{DrapeauTangpi, tangpi2014} and apply to larger class of generators than backward stochastic differential equations does.

\paragraph{Literature Discussion}
Utility optimization problems in continuous time are popular topics in finance.
\citet{karatzas1987} considered the optimization of the expected discounted utility of both consumption and terminal wealth in the complete market where they obtained an optimal consumption and wealth processes explicitly.
Using duality methods, \citet{cvitanic2001} characterized the problem of utility maximization from terminal wealth of an agent with a random endowment process in semi-martingale model for incomplete markets.
Backward stochastic differential equations, introduced in the seminal paper by \citet{peng1990} in the Lipschitz case and \citet{kobylanski2000} for the quadratic one, have revealed to be central in stating and solving problems in finance, see \citet{karoui1997}.
\citet{duffie1992} defined the concept of recursive utility by means of backward stochastic differential equations, generalized in \citet{epstein2002} and \citet{quenez2003}.
Utility optimization characterization in that context has been treated in \citet{quenez2001} in terms of a forward backward system of stochastic differential equations.
Using a martingale argumentation, \citet{imkeller2005} characterized utility maximization by means of quadratic backward stochastic differential equations for small traders in incomplete financial markets with closed constraints.
Following this line with a general utility function, \citet{horst2014} characterized the optimal strategy via a fully-coupled forward backward stochastic differential equation.
With a similar characterization, \citet{santacroce2014} considered the problem with a terminal random liability when the underlying asset price process is a continuous semi-martingale.
\citet{schweizer2007} studied a stochastic control problem arising in utility maximization under probability model uncertainty given by the relative entropy, see also \citet{schied2007}, \citet{matoussi2015}.
Backward stochastic differential equations, can be viewed themselves as generalized utility operators -- so called $g$-expectations introduced by \citet{peng1997} -- which are related to risk measures, \citet{gianinphd2002}, \citet{peng2004}, \citet{gianin2007}.
As in the classical case, maximal sub-solutions of concave backward stochastic differential equations are nonlinear expectations as well.
In this respect, \citet{tangpi2016} consider utility optimization in that framework, providing existence of optimal strategy using duality methods as well as existence of gradients.
However they do not provide a characterization of the optimal solution to which this work is dedicated to.

\paragraph{Discussion of the results and outline of the paper}
The existence and uniqueness of maximal sub-solutions in \citep{heyne2013, heyne2014, mainberger2016} depends foremost on the integrability of the terminal condition $F$, admissibility conditions on the local martingale part, and the properties of the generator -- positive, lower semi-continuous, convex in $z$ and monotone in $y$ or jointly convex in $(y,z)$.
In the present context though, the generator can no longer be positive, even uniformly linearly bounded from below.
Therefore we had to adapt the admissibility conditions, adequate for the optimization problem we are looking at.
Henceforth, we provide existence and uniqueness of maximal sub-solutions under these new admissibility conditions in Section \ref{sec:01}.
We further present there the formulation of the utility maximization problem and the transformation leading the forward backward system \eqref{eq:max-subsol}.
With this result at hand, we can address in Section \ref{sec:02} the characterization in terms of optimization of maximal sub-solutions of the forward backward stochastic differential equation.
Our first main result, Theorem \ref{thm:maintheorem}, provides a verification argument for solutions of coupled forward backward stochastic differential equation in terms of optimal strategy.
The resulting system excerpt an auxiliary backward stochastic differential equation specifying the gradient dynamic.
The second main result, Theorem \ref{thm:maintheorem02}, provides a characterization of optimal strategies in terms of solution of a coupled forward backward stochastic differential equation.
It turns out, that an auxiliary backward stochastic differential equation is necessary in order to specify the gradient of the solution.
These result extends the ones from \citet{horst2014} stated for utility maximization \`a la \citet{savage1972}.
We illustrate the results in Section \ref{sec:03} by considering utility optimization in a financial context with explicit solutions in given examples.
These explicit solutions allow to address for instance the cost of incompleteness in a financial market.
Finally, we address how the result can be applied when considering optimization for recursive utilities \`a la \citet{kreps1978} or for the present case in continuous time \`a la \citet{duffie1992}.
The proof of existence and uniqueness of maximal sub-solutions being using the same techniques as \citep{heyne2013} is postponed in Appendix \ref{sec:app01}.

\subsection{Notations}
Let $T>0$ be a fixed time horizon and $(\Omega,\mathcal{F},\left(\mathcal{F}_t\right)_{t \in \left[ 0,T \right]},P)$ be a filtered probability space, where the filtration $\left( \mathcal{F}_t\right)$ is generated by a $d$-dimensional Brownian motion $W$ and fulfills the usual conditions.
We further assume that $\mathcal{F}=\mathcal{F}_{T}$.
Throughout, we split this $d$ dimensional Brownian motion into two parts $W=(\hat{W},\tilde{W})$ with $\hat{W}=(W^1,\ldots,W^n)$ and $\tilde{W}=(W^{n+1},\ldots, W^d)$ where $1\leq  n\leq d$.
We denote by $L^0$ the set of $\mathcal{F}_T$-measurable random variables identified in the $P$-almost sure sense.
Every inequality between random variables is to be understood in the almost sure sense.
Furthermore as in the introduction, to keep the notational burden as minimal as possible, we do not write the index in $t$ and $\omega$ for the integrands unless necessary.
We furthermore generically use the short writing $\int_{}^{} \cdot $ for the process $t\mapsto \int_0^t \cdot$.
We say that a \cadlag\, process $X$ is integrable if $X_t$ is integrable for every $0\leq t\leq T$.
We use the notations
\begin{itemize}[fullwidth]
    \item $x\cdot y =\sum x_k y_k$, $x^2=x\cdot x$ and $|x|=\sqrt{x\cdot x}$ for $x$ and $y$ in $\mathbb{R}^d$.
    \item $\mathbb{R}^d_+:=\left\{ x \in \mathbb{R}^d \colon x_k\geq 0 \text{ for all }k \right\}$ and $\mathbb{R}^d_{++}:=\left\{ x \in \mathbb{R}^d \colon x_k> 0 \text{ for all }k \right\}$.
    \item for $x$ and $y$ in $\mathbb{R}^d$, let $xy:=(x_1 y_1, \ldots, x_d y_d)$ and $x/y=(x_1/y_1,\ldots,x_d/y_d)$ if $y$ is in $\mathbb{R}^d_{++}$.
    \item for $x \in \mathbb{R}^m$, $y \in \mathbb{R}^n$ and $A\in \mathbb{R}^{m\times n}$
        \begin{equation*}
            x\cdot A \cdot y:=
            \begin{bmatrix}
                x_1 & \ldots & x_m
            \end{bmatrix}
            \begin{bmatrix}
                a_{11} &\ldots & a_{1n}\\
                \vdots & \ddots &\vdots\\
                a_{m1} & \ldots& a_{mn}
            \end{bmatrix}
            \begin{bmatrix}
               y_1\\
               \vdots\\
               y_n
            \end{bmatrix}
        \end{equation*}
    \item $L^0$ and $L^p$ are the set of measurable and $p$-integrable random variables $X$ identified in the $P$-almost sure sense, $1\leq p\leq \infty$.
    \item $\mathcal{S}$ the set of \cadlag adapted processes.
    \item $\mathcal{L}$ the set of $\mathbb{R}^d$-valued predictable processes $Z$ such that $\int_{}^{} Z \cdot dW$ is a local martingale.\footnote{That is $\int_{0}^{T}Z^2 dt<\infty$ $P$-almost surely.}
    \item $\mathcal{H}$ the set of local martingales $\int_{}^{} Z\cdot dW$ for $Z \in \mathcal{L}$.
    \item $\mathcal{L}^p$ the set of those $Z$ in $\mathcal{L}$ such that $\norm{Z}_{\mathcal{L}^p}:=E[(\int_0^T Z^2dt)^{p/2}]^{1/p}<\infty$, $1\leq p<\infty$.
    \item $\mathcal{H}^p$ the set of martingales $\int_{}^{} Z \cdot dW$ for $Z \in \mathcal{L}^p$.
    \item $bmo$ the set of those $Z$ in $\mathcal{L}$ such that $\int_{}^{} Z\cdot dW$ is a bounded mean oscillations martingale.
        That is, $\|Z\|_{bmo}:=\sup_{\tau} \norm{E[ |  \int_\tau^T Z \cdot dW| |\mathcal{F}_{\tau} ]}_{\infty}<\infty$ where $\tau$ runs over all stopping times.
        Note that according to the \citep{Kazamaki01}, the $bmo_p$ norms are all equivalent for $1\leq p < \infty$ where $\|Z\|_{bmo_p}:=\sup_{\tau} \norm{E[ |  \int_\tau^T Z \cdot dW|^p |\mathcal{F}_{\tau} ]^{1/p}}_{\infty}<\infty$ where $\tau$ runs over all stopping times.
        In particular $\|Z\|_{bmo_2}=\sup_{\tau}\|E[\int_{\tau}^{T}Z^2 ds|\mathcal{F}_{\tau}]^{1/2}\|_{\infty}$.
    \item $BMO$ the set of those $\int_{}^{} Z \cdot W$ such that $Z$ is in $bmo$.
    \item $\mathcal{D}$ the set of those uniformly bounded $b \in \mathcal{L}$.
    \item $M^c$ the stochastic exponential of $c$, that is $M^c=\exp( -\int c\cdot dW-\frac{1}{2}\int c^2 dt)$.
    \item $D^b$ the stochastic discounting of $b$, that is $D^b=\exp(-\int_{}^{} b ds)$.
    \item $M^{bc}=D^bM^c=\exp(-\int_{}^{}(b+c^2/2)dt-\int_{}^{} c\cdot dW)$.
    \item For $c$ in $bmo$ we denote by $P^c$ the equivalent measure to $P$ with density
        \begin{equation*}
            \frac{dP^c}{dP} = \exp\left( -\int_{0}^{T}c \cdot dW - \frac{1}{2}\int_{0}^{T} c^2 dt   \right)
        \end{equation*}
        under which $W^c:=W + \int_{}^{} c dt$ is a Brownian motion.
    \item We generically use the notation $x=(\hat{x},\tilde{x})$ for the decomposition of vectors in $\mathbb{R}^d$ into their $n$ first components and $d-n$ last ones.
        We use the same conventions for the space $\mathcal{L}=(\hat{\mathcal{L}},\tilde{\mathcal{L}})$ where $Z=(\hat{Z},\tilde{Z})\in \mathcal{L}$.
        Also the same for $\mathcal{H}=(\hat{\mathcal{H}},\tilde{\mathcal{H}})$, $\mathcal{H}^p=(\hat{\mathcal{H}}^p,\tilde{\mathcal{H}}^p)$, $bmo=(\hat{bmo},\tilde{bmo})$ and $BMO=(\hat{BMO},\tilde{BMO})$.

        In the case where $n=d$ everything in the following with a $\,\tilde{{}}\,$ disappears or equivalently is set to $0$ and everything with a $\,\hat{{}}\,$ becomes without $\,\hat{{}}$.
\end{itemize}
For a convex function $f:\mathbb{R}^l\to (-\infty, \infty]$, we denote $f^\ast$ its convex conjugate
\begin{equation*}
    f^\ast(y)=\sup\left\{ y \cdot x - f(x)\colon x \in \mathbb{R}^l\right\}, \quad y \in \mathbb{R}^l
\end{equation*}
and denote by $\partial_{x^\ast} f$ the sub-gradients of $f$ at $x^\ast$ in $\mathbb{R}^l$, that is, the set of those $y$ in $\mathbb{R}^l$ such that $f(x)-f(x^\ast)\geq y\cdot (x-x^\ast)$ for all $x$ in $\mathbb{R}^l$.
For any $y$ in $\partial_{x^\ast}f$, it follows from classical convex analysis, see \citep{rockafellar1970}, that
\begin{equation}\label{eq:gradientopti}
    \begin{cases}
        f(x) \geq y\cdot x-f^\ast(y), &\text{for every }x\in \mathbb{R}^l \\
        \\
        f(x^\ast) = y \cdot x^\ast -f^\ast(y).
    \end{cases}
\end{equation}
If the sub-gradient is a singleton -- as in this paper -- it is a gradient and we simplify the notation to $\partial {\color{blue}f}(x^\ast)$.

\section{Maximal Sub-Solutions of FBSDEs and Utility}\label{sec:01}

A function $g:\Omega\times [0,T]\times\mathbb{R}\times \mathbb{R}^{d}\to (-\infty,\infty]$ is called a \emph{generator} if it is jointly measurable, and $g(y,z)$ is progressively measurable for any $(y,z)\in \mathbb{R}\times \mathbb{R}^{d}$.\footnote{To prevent an overload of notations, we do not mention the dependence on $\omega$ and $t$, that is, $g(y,z):=g_t(\omega,y,z)$.}
A generator is said to satisfy condition \textsc{(Std)} if
\begin{enumerate}[label=\textsc{(Std)},leftmargin=40pt]
    \item $(y, z)\mapsto g(y, z)$ is lower semi-continuous, convex with non-empty interior domain and gradients\footnote{Note that we could work with non-empty sub-gradients where by means of \citep[Theorem 14:56]{rockafellar1998} we could apply measurable selection theorem, see \citep[Corollary 1C]{gianin2007} to select measurable gradients in the sub-gradients of $g$ and working with them.} everywhere on its domain (for every $\omega$ and $t$).
        \label{std}
\end{enumerate}
\begin{remark}\label{remark}
    Note that if $g$ satisfies the above assumptions, as a normal integrand, for every $(y_0,z_0)$ in the domain of $g$ and for every $t$ and $\omega$, there exists $b$ and $c$ progressively measurable such that
    \begin{equation*}
        g(y_1,z_1) - g(y_0,z_0):=g_t(\omega, y_1, z_1) - g_t(\omega, y_0, z_0)\geq b_t(\omega) (y_1-y_0)+c_t(\omega) \cdot (z_1- z_0)
    \end{equation*}
    for every $y$, and $z$, see \citet[Chapter 14, Theorem 14.46]{rockafellar1998}.
    These processes $b$ and $c$ are the partial derivatives of $g$ with respect to $y$ and $z$, respectively.
\end{remark}
We further denote by
\begin{equation*}
    \mathcal{P}^g := \left\{ (b,c)\in \mathcal{D}\times bmo \colon E\left[ \int_{0}^{T}M^{bc}g^\ast(b,c)dt \Big | \mathcal{F}_{\cdot} \right]  \in \mathcal{H}^1\right\}.
\end{equation*}
For any \emph{terminal condition} $F$ in $L^{0}$, we call a pair $(Y,Z)$ where $Y \in \mathcal{S}$ and $Z \in \mathcal{L}$ a \emph{sub-solution} of the backward stochastic differential equation if \footnote{Note that the value process $Y$ of a sub-solution is a-priori \cadlag\, hence can jump upwards at time $T$. Therefore, when looking at maximal sub-solutions, considering sub-solutions with random endowment $Y_T\leq F$ as done in \citep{heyne2013} is equivalent to $Y_T= F$, as the latter is larger.}
\begin{equation}\label{eq:03}
    \begin{cases}
        Y_s &\displaystyle \leq Y_t-\int_{s}^{t} g(Y,Z)du-\int_{s}^{t} Z \cdot dW, \quad 0\leq s\leq t\leq T\\
        Y_T &\displaystyle = F
    \end{cases}
\end{equation}

The processes $Y$ and $Z$ are called the \emph{value} and \emph{control} processes, respectively.
Sub-solutions are not unique.
Indeed, $(Y,Z)$ is a sub-solution if and only if there exists an adapted c\`adl\`ag increasing process $K$ with $K_0=0$ such that
\begin{equation*}
    Y_t=F-\int_{t}^{T}g(Y,Z)ds-(K_T-K_t)-\int_{t}^{T}Z\cdot dW
\end{equation*}
which is given by
\begin{equation}\label{eq:increasingprocess}
    K_t=Y_t-Y_0-\int_{0}^{t} g(Y,Z)ds-\int_{0}^{t}Z\cdot dW.
\end{equation}

As mentioned in the introduction, existence and uniqueness of a maximal sub-solution depend foremost on the integrability of the positive part of $F$, admissibility conditions on the local martingale part, and the properties of the generator -- positivity, lower semi-continuity, convexity in $z$ and monotonicity in $y$ or joint convexity in $(y,z)$.
In this paper though, we removed the condition on the generator in terms of positivity to the optimization problem we are looking at.
In order to guarantee the existence and uniqueness of a maximal sub-solution, we need the following admissibility condition.
\begin{definition}
    A sub-solution $(Y,Z)$ to \eqref{eq:03} is called \emph{admissible} if $\int_{}^{} M^{bc}( Z-Yc)\cdot dW$ is in $\mathcal{H}^1$ for every $(b,c)$ in $\mathcal{P}^g$.
\end{definition}
Given a \emph{terminal condition} $F$, we denote by
\begin{equation}\label{eq:acceptance01}
    \mathcal{A}(F):=\mathcal{A}(F,g) =\left\{ (Y,Z)\in \mathcal{S}\times \mathcal{L}\colon (Y,Z)\text{ is an admissible sub-solution of }\eqref{eq:03} \right\}
\end{equation}
the set of admissible sub-solutions of \eqref{eq:03}.
A sub-solution $(Y^\ast, Z^\ast)$ in $\mathcal{A}(F)$ is called \emph{maximal sub-solution} if $Y^\ast_t\geq Y_t$ for every $0\leq t\leq T$, for every other sub-solution $(Y,Z)$ in $\mathcal{A}(F)$.
Our first result concerns the existence and uniqueness of a maximal sub-solution to \eqref{eq:03}.
\begin{theorem}\label{thm:maxsubsol}
    Let $g$ a generator satisfying \ref{std} and a terminal condition $F$ such that $E[M^{bc}_T F|\mathcal{F}_{\cdot}]$ is in $\mathcal{H}^1$ for every $(b,c)$ in $\mathcal{P}^g$.
    If $\mathcal{A}(F)$ is non empty, then there exists a unique maximal sub-solution $(Y^\ast,Z^\ast)$ in $\mathcal{A}(F)$ for which holds
    \begin{equation*}
        Y_t^\ast:=\esssup\left\{ Y_t\colon (Y,Z)\in \mathcal{A}(F) \right\}, \quad 0\leq t\leq T.
    \end{equation*}
\end{theorem}
The proof of the Theorem relies on the same techniques as in \citep{heyne2013} and is postponed into the Appendix \ref{sec:app01}.

As mentioned in the introduction, we present in a financial framework how the maximal sub-solutions are related to the utility formulation problem.
We consider a financial market consisting of one bond with interest rate $0$ and a $n$-dimensional stock price $\hat{S}=(S^1, \ldots, S^n)$ evolving according to
\begin{equation*}
    \frac{d\hat{S}}{\hat{S}}=\hat{\mu} dt +\hat{\sigma} \cdot d\hat{W}, \quad \text{and} \quad \hat{S}_0 \in \mathbb{R}^n_{++}
\end{equation*}
where $d\hat{S}/\hat{S}:=(dS^1/S^1, \ldots, dS^n/S^n)$, $\hat{\mu}$ is a $\mathbb{R}^n$-valued uniformly bounded drift process, and $\hat{\sigma}$ is a $n\times n$ volatility matrix process.
For simplicity, we assume that $\hat{\sigma}$ is invertible such that the market price of risk process
\begin{equation*}
    \hat{\theta}:=\hat{\sigma}^{-1}\cdot \hat{\mu} \quad \text{ is uniformly bounded.}
\end{equation*}
Given a $n$-dimensional \emph{trading strategy} $\hat{\eta}$, the corresponding \emph{wealth process} with initial wealth $x$ satisfies
\begin{equation*}
    X_t  = x + \int_0^t \hat{\eta} \cdot \frac{d\hat{S}}{\hat{S}} = x+\int_{0}^{t}\hat{\eta}\cdot \hat{\sigma} \cdot \hat{\theta} ds+ \int_{0}^{t} \hat{\eta}\cdot \hat{\sigma}\cdot d\hat{W} = x + \int_{0}^{t}\hat{\eta}\cdot \hat{\sigma}\cdot d\hat{W}^{\hat{\theta}}
\end{equation*}
where $\hat{W}^{\hat{\theta}}= \hat{W} + \int_{}^{} \hat{\theta}dt$ which is a Brownian motion under $P^\theta$ where $\theta = (\hat{\theta}, 0)$.
To remove the volatility factor, we generically set $\hat{\pi}=\hat{\eta} \cdot \hat{\sigma}$ and denote by $X^{\hat{\pi}}$ the corresponding wealth process.
\begin{lemma}\label{lem:integrabilityterminal}
    For every terminal condition $F$ in $L^\infty$ and $\hat{\pi}$ in $\hat{bmo}$, it holds that $E[M^{bc}_T(F+X^{\hat{\pi}}_T)|\mathcal{F}_{\cdot}]$ is in $\mathcal{H}^1$ where
    \begin{equation*}
        X^{\hat{\pi}}_t = x + \int_{0}^{t}\hat{\pi}\cdot \hat{\theta} ds + \int_{0}^{t}\hat{\pi} \cdot d\hat{W}.
    \end{equation*}
\end{lemma}
\begin{proof}
    Since $c$ is in $bmo$, it follows from reverse H\"{o}lder inequality, see \citep[Theorem 3.1]{Kazamaki01}, that there exists $p>1$ such that $E[(M^c_T)^p]<\infty$.
    Since $F$ is bounded and $\int_{}^{} \hat{\pi}\cdot d\hat{W}$ is a BMO martingale, we only need to show that $\int_{0}^{T}\hat{\pi}\cdot \hat{\theta}ds$ is in $L^q$ for all $q>1$.
    From $\hat{\pi}$ in $bmo$, it follows that $\int\hat\pi dW$ is in $\mathcal{H}^q$, for all $q>1$.
    Since $\hat{\theta}$ is uniformly bounded, for any $q>1$, it holds that
    \begin{equation*}
        E\left[\left(\int_0^T\left|\hat{\pi}\cdot\hat{\theta}\right|ds\right)^q\right] \le CE\left[\left(\int_0^T\left|\hat{\pi}\right|^2ds\right)^{q/2}\right]<\infty
    \end{equation*}
    for some constant $C$.
    Therefore, it follows from Doob's inequality that
    \begin{equation*}
        E\left[\left(\sup_{0\leq t\leq T} \left| E\left[ M^{bc}_T \left(F+X^{\hat{\pi}}_T\right) \big| \mathcal{F}_{t} \right] \right| \right)^p\right] \leq \left(\frac{p}{p-1}\right)^p E\left[\left| M^{bc}_T \left(F+X^{\hat{\pi}}_T\right)\right|^p\right] < \infty.
    \end{equation*}
\end{proof}
Given therefore a terminal condition $F$ in $L^\infty$, for every $\hat{\pi}$ in $\hat{bmo}$, according to Theorem \ref{thm:maxsubsol} together with Lemma \ref{lem:integrabilityterminal}, it follows that if $\mathcal{A}(F+X^{\hat{\pi}}_T)$ is non-empty, then there exists a unique maximal sub-solution to the forward backward stochastic differential equation
\begin{equation}\label{eq:subsolution01}
    \begin{cases}
        X^{\hat{\pi}}_t & = \displaystyle x + \int_{0}^{t}\hat{\pi}\cdot \hat{\theta}ds+ \int_{0}^{t} \hat{\pi}\cdot d\hat{W}\\
        Y_s & \leq \displaystyle Y_t - \int_{s}^{t}g(Y, Z)du -\int_{s}^{t}Z \cdot dW, \quad 0\leq s\leq t\leq T\\
        Y_T & = F+X_T^{\hat{\pi}}
    \end{cases}
\end{equation}
We denote by $U(F+X^{\hat{\pi}})$ the value of this maximal sub-solution at time $0$, and convene that if $\mathcal{A}(F+X^{\hat{\pi}})$ is empty, then $U(F)=-\infty$.
It follows from the same argumentation as in \citep{heyne2013, mainberger2016, DrapeauTangpi}, that $U$ is a concave increasing functional and therefore a utility operator\footnote{Furthermore, if $g$ is increasing in $y$, then it satisfies the sub-cash additivity property, namely $U(F-m)\geq U(F)-m$ for every $m\geq 0$. A property introduced and discussed in \citep{ravanelli2009}.}.
\begin{remark}\label{rem:expected01}
    It is known, see \citep[Example 2.1]{tangpi2016}, that -- under some adequate smoothness conditions -- the certainty equivalent $U(F)=u^{-1}(E[u(F)])$ can be described as the value at $0$ of the maximal sub-solution of the backward stochastic differential equation
    \begin{equation*}
        \begin{cases}
            Y_s\leq \displaystyle Y_t-\int_{s}^{t}\left(-\frac{u^{\prime\prime}(Y)}{2u^{\prime}(Y)}\right)Z^2du-\int_{s}^{t}Z\cdot dW, \quad 0\leq s\leq t\leq T\\
            Y_T = F
        \end{cases}
    \end{equation*}
    where $(y,z)\mapsto g(y,z)=-(u^{\prime \prime}(y) z^2)/(2u^{\prime}(y))$ is a positive jointly convex generator in many of the classical cases.
    For instance, for $u(x)=\exp(-x)$, $g(y,z)=z^2/2$, and for $u(x)=x^{r}$ with $r\in (0,1)$ and $x>0$, it follows that $g(y,z)=(1-r)z^2/(2y)$ for $(y,z)\in (0,\infty)\times \mathbb{R}^d$.
\end{remark}
Before we proceed to characterization of optimal strategies, let us point to a simple transformation that underlies the following section.
For $(Y,Z)$ sub-solution in $\mathcal{A}(F+X_T^{\hat{\pi}})$, the variable change $\bar{Y} := Y - X^{\hat{\pi}}$, $\bar{Z} := Z-\pi$ where $\pi=(\hat{\pi},0)$ leads to the following system of forward backward stochastic differential equation
\begin{equation}\label{eq:subsolution02}
    \begin{cases}
        X_t & = \displaystyle x + \int_{0}^{t} \hat{\pi}\cdot \hat{\theta}ds+\int_{0}^{t}\hat{\pi}\cdot d\hat{W}\\
        \bar{Y}_s & \leq \displaystyle \bar{Y}_t - \int_{s}^{t}\left[g(\bar{Y}+X, \bar{Z} + \pi) - \pi\cdot \theta\right] du -\int_{s}^{t}\bar{Z} \cdot dW, \quad 0\leq s\leq t\leq T\\
        \bar{Y}_T & = F
    \end{cases}
\end{equation}
where $\theta = (\hat{\theta}, 0)$.
In the following we consistently use the notation $\bar{Y} = Y-X$ and $\bar{Z}=Z-\pi$ where $(Y,Z)$ is sub-solution of the utility problem.

\section{Sufficient Characterization of the Coupled FBSDE System}\label{sec:02}

We are interested in a utility maximization problem with random endowment $F$ in $L^\infty$, for the utility function $U$.
In other terms, finding $\hat{\pi}^\ast$ in $\hat{bmo}$ such that
\begin{equation}\label{eq:problem}
    U(F+X^{\hat{\pi}^\ast}_T) \geq U\left( F+X^{\hat{\pi}}_T \right) \quad \text{for all trading strategy }\hat{\pi} \in \hat{bmo}.
\end{equation}
We call such a strategy $\hat{\pi}^\ast$ an \emph{optimal strategy} to problem \eqref{eq:problem}.
Throughout, we call any trading strategy $\hat{\pi}$ in $\hat{bmo}$ an \emph{admissible strategy}.
We split this Section into two, namely a verification result and a characterization result in the spirit of \citep{horst2014} which has been done in the context of classical expected utility optimization.

\subsection{Verification}
Our first main result is a verification theorem for the optimal solution given by the solution of a fully coupled backward stochastic differential equation.

\begin{theorem}\label{thm:maintheorem}
    Suppose that there exists $\eta(y,z,\hat{v}):=(\hat{\eta}(y,z,\hat{v}),0)$ such that
    \begin{equation}\label{eq:sufficientcond}
        \partial_{\hat{z}} g(y, z+\eta(y,z,\hat{v}))=\hat{v}+\hat{\theta}
    \end{equation}
    Suppose that the fully coupled forward backward system of stochastic differential equations
    \begin{equation}\label{eq:sufficient}
        \begin{cases}
            X_t & \displaystyle = x+\int_0^t \hat{\eta}(X+\bar{Y},\bar{Z},\hat{V})\cdot \hat{\theta}ds+\int_{0}^{t} \hat{\eta}(X+\bar{Y},\bar{Z},\hat{V})\cdot d\hat{W}\\
            \bar{Y}_t & \displaystyle = F-\int_{t}^{T}\left\{ g\left(X+\bar{Y},\bar{Z}+\eta\left( X+\bar{Y},\bar{Z},\hat{V} \right)\right)-\hat{\eta}(X+\bar{Y},\bar{Z},\hat{V})\cdot \hat{\theta}\right\}ds-\int_{t}^{T}\bar{Z} \cdot d W\\
            U_t & \displaystyle = U_T+\int_{t}^{T} \left(\frac{\hat{V}^2-\tilde{V}^2}{2}+\hat{V}\cdot \hat{\theta}\right) ds +\int_{t}^{T}V\cdot d W\\
            U_T &=\displaystyle \int_{0}^{T} \left(\partial_yg\left( X+\bar{Y}, \bar{Z}+\eta(X+\bar{Y},\bar{Z},\hat{V}) \right)+\frac{\partial_{\tilde{z}}g\left( X+\bar{Y},\bar{Z}+\eta(X+\bar{Y},\bar{Z},\hat{V})\right)^2}{2}\right)ds\\
                & \displaystyle \quad\quad\quad +\int_{0}^{T}\partial_{\tilde{z}}g\left( X+\bar{Y},\bar{Z}+\eta(X+\bar{Y},\bar{Z},\hat{V}) \right)\cdot d\tilde{W}
        \end{cases}
    \end{equation}
    admits a solution $(X^\ast,\bar{Y}^\ast,\bar{Z}^\ast,U,V)$ such that
    \begin{itemize}
        \item $\hat{\pi}^\ast:=\hat{\eta}(X^\ast+\bar{Y}^\ast,\bar{Z}^\ast,\hat{V})$ is in $\hat{bmo}$;
        \item $(Y^\ast, Z^\ast):=(\bar{Y}^\ast + X^\ast, \bar{Z}^\ast+\pi^\ast)$ satisfies $\int M^{bc}\left( Z^\ast - Y^\ast c \right)\cdot dW$ is in $\mathcal{H}^1$ for every $(b,c)$ in $\mathcal{P}^g$.
        \item $(b^\ast, c^\ast)$ is in $\mathcal{P}^g$ where $b^\ast:=\partial_yg(Y^\ast,Z^\ast)$ and $c^\ast :=\partial_{z} g(Y^\ast,Z^\ast)$;
    \end{itemize}
    Then, $\hat{\pi}^\ast$ is an optimal strategy to problem \eqref{eq:problem} and
    \begin{equation*}
        U(F+X_T^{\hat{\pi}^\ast})=E\left[ M_T^{b^\ast c^\ast}\left( F+X_T^{\hat{\pi}^\ast} \right) + \int_{0}^{T}M^{b^\ast c^\ast} g^\ast(b^\ast,c^\ast)dt  \right] = \bar{Y}^\ast_0+x.
    \end{equation*}
\end{theorem}
\begin{remark}
    The conditions on the gradient \eqref{eq:sufficientcond} together with the auxiliary BSDE in $(U,V)$ guarantees that the measure with density $M^{b^\ast c^\ast}_T$ is orthogonal to the linear space $\{X^{\hat{\pi}}_T:\hat{\pi} \in \hat{bmo}\}$ of the wealth processes generated by the strategies $\hat{\pi}$ in $\hat{bmo}$.
    Indeed, the auxiliary BSDE in $(U,V)$ is related to an orthogonal projection in terms of measure.
\end{remark}
Before addressing the proof of the theorem, let us show the following lemma concerning the auxiliary BSDE in $(U,V)$ characterizing the gradient of the optimal solution.
\begin{lemma}\label{lem:linearbsde}
    Let $b \in \mathcal{D}$ and $\tilde{c} \in \tilde{bmo}$.
    The backward stochastic differential equation
    \begin{equation*}
        \begin{cases}
            U_t & = \displaystyle U_T+\int_{t}^{T}\left(\frac{\hat{V}^2-\tilde{V}^2}{2}+\hat{V}\cdot \hat{\theta}\right)ds+\int_{t}^{T}V \cdot dW\\
            U_T & = \displaystyle \int_{0}^{T}\left(b +\frac{\tilde{c}^2}{2} \right)dt+\int_{0}^{T}\tilde{c} \cdot d\tilde{W}
        \end{cases}
    \end{equation*}
    admits a unique solution $(U,V)$ with $V$ in $bmo$.
    In this case, if we define $c=(\hat{V}+\hat{\theta},\tilde{c})$ which is in $bmo$, it follows that
    \begin{equation}\label{eq:M_T}
        M_{T}^{bc}=\exp\left( -\int_{0}^{T}\left(\frac{\tilde{V}^2}{2} +\frac{\hat{\theta}^2}{2} \right)dt+\int_{0}^{T}\tilde{V} \cdot d\tilde{W} -\int_0^T\hat{\theta} \cdot d\hat{W} -U_0 \right).
    \end{equation}
\end{lemma}
\begin{proof}
    According to \citet{kobylanski2000}, since $\int_{0}^{T}b ds$ is uniformly bounded, the backward stochastic differential equation
    \begin{equation*}
        Y_t=\int_{0}^Tbds+\int_{t}^T\left(\frac{\hat{Z}^2}{2}-\frac{\tilde{Z}^2}{2}\right)ds+\int_t^T Z \cdot dW^{(\hat{\theta},\tilde{c})}
    \end{equation*}
    admits a unique solution $(Y,Z)$ where $Y$ is uniformly bounded and $Z$ is in $\mathcal{L}^2(P^{(\hat{\theta},\tilde{c})})$.
    According to \citet[Proposition 2.1]{Briand2013} it also holds that $Z$ is in $bmo(P^{(\hat{\theta},\tilde{c})})$ which is also in $bmo$ since $(\hat{\theta},\tilde{c})$ is in $bmo$.\footnote{The $bmo$ space is invariant under $bmo$ measure change, see \citep[Theorem 3.6]{Kazamaki01}.}
    The variable change $U=Y+\int\tilde{c}^2/2dt+\int \tilde{c}\cdot d\tilde{W}$ and $V=(\hat{Z},\tilde{Z}-\tilde{c})$ which is in $bmo$ yields
    \begin{align*}
        U_t&=U_T+\int_t^T\left(\frac{\hat{Z}^2}{2}+\hat{\theta}\cdot \hat{Z}+\tilde{c}\cdot \tilde{Z}-\frac{\tilde{Z}^2}{2}-\frac{\tilde{c}^2}{2}\right)ds+\int_t^T\hat{Z}\cdot d\hat{W}+\int_t^T\left(\tilde{Z}-\tilde{c}\right)\cdot d\tilde{W}\\
           &=U_T+\int_{t}^{T}\left(\frac{\hat{V}^2-\tilde{V}^2}{2}+\hat{V}\cdot \hat{\theta}\right)ds+\int_{t}^{T}V \cdot dW
    \end{align*}
    showing the first assertion.
    Defining now $c=(\hat{V}+\hat{\theta},\tilde{c})$, which is in $bmo$, it follows that
    \begin{multline*}
        -\int_{0}^{T}\left(b+\frac{c^2}{2} \right)dt-\int_{0}^{T}c \cdot dW\\
        =-\int_{0}^{T}\left(b+\frac{\tilde{c}^2}{2}\right) dt-\int_{0}^{T}\tilde{c}\cdot d\tilde{W}-\int_{0}^{T} \frac{\left(\hat{V}+\hat{\theta}\right)^2}{2}dt -\int_{0}^{T}\left(\hat{V}+\hat{\theta}\right)\cdot d\hat{W} \\
        =-U_0+\int_{0}^{T}\left(\frac{\hat{V}^2 -\tilde{V}^2}{2} + \hat{V}\cdot \hat{\theta}\right)ds+\int_{0}^{T}V \cdot dW-\int_{0}^{T} \frac{\left(\hat{V}+\hat{\theta}\right)^2}{2}dt -\int_{0}^{T}\left(\hat{V}+\hat{\theta}\right)\cdot d\hat{W}\\
        =-\int_{0}^{T}\frac{\tilde{V}^2}{2} dt-\int_0^T\frac{\hat{\theta}^2}{2}dt+\int_{0}^{T}\tilde{V}\cdot d\tilde{W}-\int_0^T\hat{\theta}\cdot d\hat{W}-U_0.
    \end{multline*}
    Taking the exponential on both sides, yields \eqref{eq:M_T}.
\end{proof}
With this Lemma at hand, we are in position to address the proof of Theorem \ref{thm:maintheorem}.
\begin{proof}[Theorem \ref{thm:maintheorem}]
    Let $\hat{\pi}^\ast:=\hat{\eta}(X^\ast+\bar{Y}^\ast,\bar{Z}^\ast,\hat{V})$ where $(X^\ast,\bar{Y}^\ast,\bar{Z}^\ast,U,V)$ is a solution of \eqref{eq:sufficient}.
    Let further $Y^\ast = \bar{Y}^\ast + X^\ast$, $Z^\ast = \bar{Z}^\ast + \pi^\ast$ where $\pi^\ast = (\hat{\pi}^\ast, 0)$ and adopt the notations $b^\ast = \partial_y g(Y^\ast, Z^\ast)$, $c^\ast=\partial_{z}g(Y^\ast, Z^\ast)$.
    Since $X^\ast = X^{\hat{\pi}^\ast}$, it follows that $(Y^\ast,Z^\ast)$ satisfies
    \begin{equation*}
        \begin{cases}
            X^{\hat{\pi}^\ast} &= \displaystyle x + \int_{0}^{t}\hat{\pi}^\ast \cdot \hat{\theta}ds + \int_{0}^{t}\hat{\pi}^\ast \cdot d\hat{W}\\
            Y_s^\ast & = \displaystyle Y_t^\ast - \int_{s}^{t}g(Y^\ast, Z^\ast)du -\int_{s}^{t}Z^\ast \cdot dW, \quad 0\leq s\leq t\leq T\\
            Y_T^\ast & = F+X_T^{\hat{\pi}^\ast}
        \end{cases}
    \end{equation*}
    Now by assumption, $\hat{\pi}^\ast$ is in $\hat{bmo}$ and $\int M^{bc}\left( Z^\ast - Y^\ast c \right)\cdot dW$ is in $\mathcal{H}^1$ for every $(b,c)$ in $\mathcal{P}^g$.
    We deduce that $(Y^\ast, Z^\ast)$ is in $\mathcal{A}(F+X^{\hat{\pi}^\ast}_T)$ and therefore $Y^\ast_0\leq U(F+X^{\hat{\pi}^\ast}_T)$.
    Furthermore, since $(b^\ast, c^\ast)= (\partial_yg(Y^\ast, Z^\ast), \partial_z g(Y^\ast, Z^\ast))$, according to \eqref{eq:gradientopti}, it follows that $g(Y^\ast, Z^\ast) = b^\ast Y^\ast + {\color{red}c^*}\cdot Z^\ast - g^\ast(b^\ast, c^\ast)$.
    Hence
    \begin{equation*}
        Y^\ast_t = F+X_T^{\hat{\pi}^\ast} - \int_{t}^{T}\left( b^\ast Y^\ast + c^\ast\cdot Z^\ast -g^\ast(b^\ast, c^\ast)\right)ds - \int_{t}^{T}Z^\ast \cdot dW.
    \end{equation*}
    Since $(b^\ast, c^\ast)$ is in $\mathcal{P}^g$ and $\hat{\pi}^\ast$ in $\hat{bmo}$ we deduce that
    \begin{equation*}
        Y^\ast_0 = E\left[ M^{b^\ast c^\ast}\left( F+X^{\hat{\pi}^\ast} \right) + \int_{0}^{T}M^{b^\ast c^\ast}g^\ast(b^\ast, c^\ast)dt   \right].
    \end{equation*}

    As for the rest of the theorem, since $(Y^\ast, Z^\ast)$ is in $\mathcal{A}(F+X_T^{\hat{\pi}^\ast})$, we are left to show that for any $\hat{\pi}$ in $\hat{bmo}$ and any $(Y,Z)$ in $\mathcal{A}(F+X_T^{\hat{\pi}})$ it follows that $Y^\ast_0\geq Y_0$.
    Indeed, it would follows that
    \begin{itemize}
        \item $Y^\ast_0 \geq Y_0$ for every $(Y,Z)$ in $\mathcal{A}(F+X_T^{\hat{\pi}^\ast})$ and therefore $Y^\ast_0 = U(F+X_T^{\hat{\pi}^\ast})$;
        \item $Y^\ast_0 \geq Y_0$ for every $(Y,Z)$ in $\mathcal{A}(F+X_T^{\hat{\pi}})$ and every $\hat{\pi}$ in $\hat{bmo}$ showing that $Y^\ast_0 \geq U(F + X_T^{\hat{\pi}})$ for every $\hat{\pi}$ in $\hat{bmo}$.
    \end{itemize}
    Let therefore $\hat{\pi}$ in $\hat{bmo}$.
    Without loss of generality we may assume that $\mathcal{A}(F+X_T^{\hat{\pi}})$ is non-empty.
    Let $(Y,Z)$ in $\mathcal{A}(F+X_T^{\hat{\pi}})$ and denote by $\Delta Y:=Y-Y^\ast$, $\Delta Z:=Z-Z^\ast$ and $\Delta \hat{\pi}=\hat{\pi} -\hat{\pi}^\ast$.
    According to Remark \ref{remark}, it follows that $g(Y,Z) - g(Y^\ast,Z^\ast) \geq b^\ast \Delta Y + c^\ast \cdot \Delta Z$.
    Hence
    \begin{align*}
        \Delta Y_s&\leq \Delta Y_t -\int_{s}^{t} \left[g(Y,Z)-g(Y^\ast,Z^\ast)\right] du-\int_{s}^{t}\Delta Z\cdot dW\\
                  &\leq \Delta Y_t-\int_{s}^{t}\left[b^\ast \Delta Y+c^\ast \cdot \Delta Z \right]du-\int_{s}^{t}\Delta Z\cdot dW.
    \end{align*}
    By the change of variable $\check{Y}=M^{b^\ast c^\ast} \Delta Y$ and $\check{Z}=M^{b^\ast c^\ast} (\Delta Z-\Delta Yc^\ast)$, it follows that $(\check{Y},\check{Z})$ satisfies\footnote{Recall, $\hat{W}^{\hat{\theta}} = \hat{W} + \int_{}^{} \hat{\theta} dt$.}
    \begin{equation*}
        \check{Y}_t \leq M^{b^\ast c^\ast}_T\left(X^{\hat{\pi}}_T-X^{\hat{\pi}^\ast}_T\right) -\int_{t}^{T} \check{Z}\cdot dW = E\left[ M^{b^\ast c^\ast}_T\left(\int_{0}^{T} \hat{\pi} \cdot d\hat{W}^{\hat{\theta}} \right) \Big |\mathcal{F}_t \right]
    \end{equation*}
    since $\int \check{Z}\cdot d W$ is a martingale as the difference of two martingales; $(b^\ast, c^\ast)$ being in $\mathcal{P}^g$.
    However, $c^\ast$ and $b^\ast$ satisfying the condition of Lemma \ref{lem:linearbsde}, it follows that $\tilde{V}$ is in particular in $\tilde{bmo}$.
    Hence, $W^{(\hat{\theta}, -\tilde{V})} := (\hat{W} + \int_{}^{} \hat{\theta} dt, \tilde{W} - \int_{}^{} \tilde{V}dt)=(\hat{W}^{\hat{\theta}}, \tilde{W} - \int_{}^{} \tilde{V}dt)$ is a Brownian motion under the measure $P^{(\hat{\theta}, -\tilde{V})}$.
    Since $\check{Y}_0 = \Delta Y_0$, for $t=0$, according to \eqref{eq:M_T}, we have
    \begin{align*}
        \Delta Y_0 = \check{Y}_0 & \leq E\left[ M_T^{b^\ast c^\ast} \int_{0}^{T} \Delta \hat{\pi} \cdot d \hat{W}^{\hat{\theta}}\right]\\
                                 & =E\left[ \exp\left( -\int_{0}^{T}\frac{\tilde{V}^2}{2}dt-\int_0^T\frac{\theta^2}{2}dt+\int_{0}^{T}\tilde{V} \cdot d\tilde{W}-\int_0^T\hat{\theta} \cdot d\hat{W}-U_0   \right) \int_0^T \Delta \hat{\pi} \cdot d \hat{W}^{\hat{\theta}} \right]\\
                                & =\exp(-U_0)E^{(\hat{\theta},-\tilde{V})}\left[\int_0^T \Delta \hat{\pi} \cdot d \hat{W}^{\hat{\theta}}\right]=0.
    \end{align*}
    Thus, $Y^\ast_0\geq Y_0$ which ends the proof.
\end{proof}
\begin{remark}
    Note that the proof of the theorem shows in particular that the maximal sub-solution for the optimal utility is $(Y^\ast, Z^\ast)$ which satisfies a ``linear''\footnote{Naturally, the coefficients $a^\ast$, $b^\ast$ and $c^\ast$ depend on $\hat{\pi}^\ast,X^\ast,\bar{Y}^\ast,\bar{Z}^\ast$, but are actually the gradients evaluated at the value of the optimal solution.} backward stochastic differential equation
    \begin{equation*}
        Y^\ast_t=F+X_T^{\hat{\pi}^\ast}-\int_{t}^{T} \left[b^\ast Y^\ast+c^\ast \cdot Z^\ast-g^\ast(b^\ast, c^\ast)\right]ds-\int_{t}^{T} Z^\ast \cdot dW.
    \end{equation*}
\end{remark}

\begin{remark}
    The case of utility optimization for the certainty equivalent $U(F)=u^{-1}(E[u(F)])$ or its equivalent formulation in terms of expected utility $E[u(F)]$ in a backward stochastic differential equation context has been the subject of several papers, in particular \citep{horst2014} and \citep{tangpi2016}.
    The optimal solutions provided in those papers each correspond to the coupled forward backward stochastic differential equation system of Theorem \eqref{thm:maintheorem}.
    Indeed, as mentioned in Remark \ref{rem:expected01}, the generator $g$ corresponds to $g(y,z)=-(u^{\prime\prime}(y)z^2)/(2u^{\prime}(y))$.
    In that context, the coupled system of forward backward stochastic differential equations in Theorem \ref{thm:maintheorem} corresponds to
\begin{equation*}
    \begin{cases}
        X_t & \displaystyle = x+\int_{0}^{t} \hat{\theta}\cdot \left(-\hat{\bar{Z}}-\frac{u^\prime(X+\bar{Y})}{u^{\prime\prime}(X+\bar{Y})}(\hat{\theta}+\hat{V})\right)ds+\int_{0}^{t} \left(-\hat{\bar{Z}}-\frac{u^\prime(X+\bar{Y})}{u^{\prime\prime}(X+\bar{Y})}(\hat{\theta}+\hat{V})\right) \cdot d\hat{W}\\
        \bar{Y}_t & \displaystyle = F-\int_{t}^{T} \left(-\left(\frac{u^{\prime}(X+\bar{Y})}{2u^{\prime\prime}(X+\bar{Y})}\left(\hat{\theta}+\hat{V}\right)^2+\frac{u^{\prime\prime}(X+\bar{Y})}{2u^{\prime}(X+\bar{Y})}\tilde{\bar{Z}}^2\right)-\hat{\theta}\cdot \left(-\hat{\bar{Z}}-\frac{u^\prime(X+\bar{Y})}{u^{\prime\prime}(X+\bar{Y})}(\hat{\theta}+\hat{V})\right) \right)ds \\
            & \quad \quad \quad \quad-\displaystyle \int_{t}^{T}\bar{Z} \cdot d W\\
        U_t & \displaystyle = U_T+\int_{t}^{T} \left(\frac{\hat{V}^2-\tilde{V}^2}{2}+\hat{\theta}\cdot \hat{V}\right) ds +\int_{t}^{T}V\cdot d W\\
        U_T &=\displaystyle \int_{0}^{T} \left(\partial_yg \left(X+\bar{Y},\bar{Z} +{\eta}(X+\bar{Y},\bar{Z},\hat{V})\right)+\frac{1}{2}\left(\frac{u^{\prime\prime}(X+\bar{Y})}{u^{\prime}(X+\bar{Y})} \tilde{\bar{Z}}\right)^2\right)ds\\
            & \quad\quad\quad\quad \displaystyle +\int_{0}^{T}-\frac{u^{\prime\prime}(X+\bar{Y})}{u^{\prime}(X+\bar{Y})}\tilde{\bar{Z}}\cdot d\tilde{W}
    \end{cases}
\end{equation*}
    It turns out that $\partial_yg \left(x+y,z +{\eta}(x+y,z,\hat{v})\right)=0$, implies in that case that $\hat{v}=0$ and therefore
    \begin{equation*}
        \begin{cases}
            \partial_yg \left(x+y,z +{\eta}(x+y,z,\hat{v})\right)& =-\displaystyle \frac{u^{\prime\prime\prime}(x+y)u^{\prime}(x+y)-\left(u^{\prime\prime}(x+y)\right)^2}{2\left(u^{\prime}(x+y)\right)^2}\left(\frac{\left(u^{\prime}(x+y)\right)^2}{\left(u^{\prime\prime}(x+y)\right)^2}\hat{\theta}^2+\tilde{z}^2\right)\\
                                                                 & = 0\\
            \partial_{\hat{z}}g(x+y,z+\eta(x+y,z,\hat{v}))& =\hat{\theta}
        \end{cases}
    \end{equation*}
    Under these conditions, the forward backward stochastic differential equation turns into
    \begin{equation*}
        \begin{cases}
            X_t & \displaystyle = x+\int_{0}^{t} \hat{\theta}\cdot\left(-\hat{\bar{Z}}-\frac{u^\prime(X+\bar{Y})}{u^{\prime\prime}(X+\bar{Y})}\hat{\theta}\right)ds+\int_{0}^{t} \left(-\hat{\bar{Z}}-\frac{u^\prime(X+\bar{Y})}{u^{\prime\prime}(X+\bar{Y})}\hat{\theta}\right) \cdot d\hat{W}\\
            \bar{Y}_t & \displaystyle = F-\int_{t}^{T} \left(-\left(\frac{u^{\prime}(X+\bar{Y})}{2u^{\prime\prime}(X+\bar{Y})}\hat{\theta}^2+\frac{u^{\prime\prime}(X+\bar{Y})}{2u^{\prime}(X+\bar{Y})}\tilde{\bar{Z}}^2\right)-\hat{\theta}\cdot \left(-\hat{\bar{Z}}-\frac{u^\prime(X+\bar{Y})}{u^{\prime\prime}(X+\bar{Y})}\hat{\theta}\right) \right)ds -\int_{t}^{T}\bar{Z} \cdot d W\\
            U_t & \displaystyle = U_T+\int_{t}^{T} \frac{-\tilde{V}^2}{2} ds +\int_{t}^{T}\tilde{V}\cdot d \tilde{W}\\
            U_T &=\displaystyle \int_{0}^{T} \frac{1}{2}\left(\frac{u^{\prime\prime}(X+\bar{Y})}{u^{\prime}(X+\bar{Y})}\tilde{\bar{Z}}^2\right)ds+\int_{0}^{T}-\frac{u^{\prime\prime}(X+\bar{Y})}{u^{\prime}(X+\bar{Y})}\tilde{\bar{Z}}\cdot d\tilde{W}
        \end{cases}
    \end{equation*}
    which coincide with the forward backward stochastic differential equation system in \citep{horst2014}, noting that the auxiliary backward stochastic differential equation in $(U,V)$ disappears by a transformation.
    For classical utility functions such as exponential with random endowment, and power or logarithmic without endowment, the optimization problem can be solved by solving quadratic backward stochastic differential equations, see \citep{imkeller2005}.
    Their method relies on a ``separation of variables'' property shared by those classical utility functions.
    In the case of exponential utility, as seen in the first case study of Section \ref{sec:03} in the case where $\beta=0$, our forward backward stochastic differential equation system reduces to a simple backward stochastic differential equation system.
\end{remark}

\subsection{Characterization}
Our second main result is a characterization theorem of optimal solutions in terms of the fully coupled system of forward backward stochastic differential equations presented in Theorem \ref{thm:maintheorem}.
\begin{theorem}\label{thm:maintheorem02}
    Suppose that $\hat{\pi}^\ast$ in $\hat{bmo}$ is an optimal strategy to problem \eqref{eq:problem}.
    Denote by $(Y^\ast, Z^\ast)$ the corresponding maximal sub-solution to problem \eqref{eq:problem} for $\hat{\pi}^\ast$ and denote $\bar{Y}^\ast := Y^\ast - X^{\hat{\pi}^\ast}$ as well as $\bar{Z}^\ast := Z^\ast -\hat{\pi}^\ast$.
    Under the assumptions
    \begin{itemize}
        \item the sub-solution $(Y^\ast, Z^\ast)$ is a solution;
        \item the concave function $\mathbb{R} \ni m \mapsto f(m):=U(F+X^{m\hat{\pi}+\hat{\pi}^\ast}_T)-U(F+X^{\hat{\pi}^\ast}_T)$, is differentiable at $0$ for every $\hat{\pi}$ in $\hat{bmo}$.
        \item $(b^\ast, c^\ast)$ is in $\mathcal{P}^g$ where $b^\ast := \partial_y g(Y^\ast, Z^\ast)$ and $c^\ast :=\partial_z g(Y^\ast, Z^\ast)$;
        \item the point-wise implicit solution $\eta(y,z,\hat{v})=(\hat{\eta}(y,z,\hat{v}),0)$ to $\partial_{\hat{z}}g(y, z+\eta(y,z,\hat{v}))=\hat{v}+\hat{\theta}$ is unique for every given $y$, $z$ and $\hat{v}$;
    \end{itemize}
    then it holds that
    \begin{equation*}
        \hat{\pi}^\ast =\hat{\eta}(X^\ast+\bar{Y}^\ast, \bar{Z}^\ast,\hat{V})\quad P\otimes dt\text{-almost surely}
    \end{equation*}
    where $(U,V)$ is the unique solution with $V$ in $bmo$ of
    \begin{equation}\label{eq:auxiliarybsde}
        \begin{cases}
        U_t & \displaystyle = U_T+\int_{t}^{T} \left(\frac{\hat{V}^2-\tilde{V}^2}{2}+\hat{V}\cdot \hat{\theta}\right) ds +\int_{t}^{T}V\cdot d W\\
            U_T &=\displaystyle \int_{0}^{T} \left(b^\ast+\frac{(\tilde{c}^\ast)^2}{2}\right)ds+\int_{0}^{T}\tilde{c}^\ast \cdot d\tilde{W}
        \end{cases}
    \end{equation}
    In particular, the fully coupled forward backward stochastic differential equation system of Theorem \ref{thm:maintheorem} has a solution $(X^{\hat{\pi}^\ast}, \bar{Y}^\ast, \bar{Z}^\ast, U, V)$.
\end{theorem}
\begin{proof}
    Let $\hat{\pi}$ in $\hat{bmo}$.
    By assumption, the function $f$ is concave, admits a maximum at $0$ and is differentiable at $0$.
    In particular, on a neighborhood of $0$, $f$ is real valued.
    For $m$ in such neighborhood, we denote by $(Y^m,Z^m)$ the maximal sub-solution in $\mathcal{A}(F+X_T^{m\hat{\pi}+\hat{\pi}^\ast})$.
    Since $(b^\ast,c^\ast)$ is in $\mathcal{P}^g$, it follows that $\int_{}^{} (M^{b^\ast c^\ast}Z^m-M^{b^\ast c^\ast}Y^mc)\cdot dW$ is a martingale.
    By the same argumentation as in the proof of Theorem \ref{thm:maintheorem}, it holds
    \begin{equation*}
        f(m)=U\left(F+X_T^{m\hat{\pi}+\hat{\pi}^\ast}\right)-U\left(F+X^{\hat{\pi}^\ast}_T\right)\leq m E\left[ M^{b^\ast c^\ast}_T \int_{0}^{T}\hat{\pi}\cdot d\hat{W}^{\theta} \right]
    \end{equation*}
    for every $m$ in a neighborhood of $0$.
    In particular $E[M^{b^\ast c^\ast}_T\int_{0}^{T}\hat{\pi}\cdot d\hat{W}^{\theta}]$ is in the sub-gradient of $f$ at $0$, which is equal to $0$ since $f$ is concave, maximal at $0$ and differentiable at $0$.
    It follows that
    \begin{equation*}
        E\left[ M^{b^\ast c^\ast}_T\int_{0}^{T}\hat{\pi} \cdot d\hat{W}^{\theta}  \right]=0 \quad \text{for all }\hat{\pi}\in \hat{bmo}.
    \end{equation*}
    Since $M=E[M^{b^\ast c^\ast}_T|\mathcal{F}_{\cdot}]$ is a strictly positive martingale in $\mathcal{H}^1$, by martingale representation theorem, it follows that
    \begin{equation*}
        M=M_0+\int_{}^{} M\hat{H} \cdot d\hat{W}+\int_{}^{} M\tilde{H}\cdot d\tilde{W}
    \end{equation*}
    for which, using the same argumentation methods as in the proof of Lemma \ref{lem:linearbsde}, $H$ is in $bmo$.
    Therefore, it holds that
    \begin{align*}
        \frac{M_t}{M_t^{\theta}} &= M_0 + \int_{0}^{t} \frac{M\hat{H}}{M^{\theta}}\cdot d\hat{W} + \int_{0}^{t} \frac{M\tilde{H}}{M^{\theta}}\cdot d\tilde{W} + \int_{0}^{t} \frac{M\hat{\theta}}{M^{\theta}}\cdot d\hat{W} + \int_{0}^{t} \frac{M\hat{\theta}^2}{M^{\theta}}dt+ \int_{0}^{t} \frac{M\hat{\theta}\cdot \hat{H}}{M^{\theta}}dt\\
        &= M_0 + \int_{0}^{t} \frac{M\left(\hat{H}+\hat{\theta}\right)}{M^{\theta}}\cdot d \hat{W}^\theta + \int_{0}^{t} \frac{M\tilde{H}}{M^{\theta}}\cdot d\tilde{W}.
    \end{align*}
    Hence
    \begin{equation*}
        0=E\left[ M_T^{b^\ast c^\ast} \int_{0}^{T}\hat{\pi}\cdot d\hat{W}^{\theta}  \right]=E^{\theta}\left[ \int_{0}^{T}\frac{M\left(\hat{H}+\hat{\theta}\right)}{M^{\theta}}\cdot \hat{\pi}dt  \right] \quad \text{for all }\hat{\pi} \in \hat{bmo}
    \end{equation*}
    showing that $\hat{H}=-\hat{\theta}$, $P\otimes dt$-almost surely and therefore
    \begin{equation*}
        M=M_0 \exp\left(\int_{}^{} \tilde{H}\cdot d\tilde{W}-\frac{1}{2}\int_{}^{} \tilde{H}^2dt -\frac{1}{2}\int_{}^{}\hat{\theta}^2dt-\int_{}^{}\hat{\theta}\cdot d\hat{W}\right).
    \end{equation*}
    Since $M_T^{b^\ast c^\ast}=M_T$, we deduce that
    \begin{multline*}
        -\int_{0}^{T}\left(b^\ast +\frac{(c^\ast)^2}{2}\right)dt-\int_{0}^{T}c^\ast \cdot dW \\
        = \ln(M_0)+ \int_{0}^{T} \tilde{H}\cdot d\tilde{W}-\frac{1}{2}\int_{0}^{T} \tilde{H}^2dt -\frac{1}{2}\int_0^T\hat{\theta}^2dt-\int_0^T\hat{\theta}\cdot d\hat{W}.
    \end{multline*}
    Defining
    \begin{equation*}
        \begin{cases}
            V & = (\hat{c}^\ast-\hat{\theta}, \tilde{H})\\
            U_t & = \displaystyle -\ln(M_0)-\int_{0}^{t} \left( \frac{(\hat{c}^\ast-\hat{\theta})^2}{2}-\frac{(\tilde{H})^2}{2} +\hat{\theta}\cdot \left(\hat{c}^\ast-\hat{\theta}\right) \right)ds-\int_{0}^{t}\left(\hat{c}^\ast-\hat{\theta}\right) \cdot d\hat{W}-\int_{0}^{t}\tilde{H}\cdot d\tilde{W}
        \end{cases}
    \end{equation*}
    shows that $(U,V)$ satisfies the auxiliary backward stochastic differential equation \eqref{eq:auxiliarybsde}, which by means of Lemma \ref{lem:linearbsde} admits a unique solution.
    Hence
    \begin{equation*}
        \hat{\theta}+\hat{V}=\hat{c}^\ast=\partial_{\hat{z}}g\left( X^\ast+Y^\ast, Z^\ast+\hat{\pi}^\ast \right) \quad P\otimes dt\text{-almost surely}
    \end{equation*}
    which by uniqueness of the point-wize solution $\hat{\eta}(y,z,\hat{v})$ implies that $\hat{\pi}^\ast=\hat{\eta}(X^\ast+Y^\ast,Z^\ast, \hat{V})$ $P\otimes dt$-almost surely.
\end{proof}
\begin{remark}
    Existence of optimal strategies $\hat{\pi}^\ast$ such that $U(F+X^{\hat{\pi}^\ast}_T)\geq U(F+X^{\hat{\pi}}_T)$ for every $\hat{\pi}$ in $\hat{bmo}$ are often showed using functional analysis and duality methods, see for instance \citep{walter2001, kramkov1999} for the case of expected utility.
    Present functionals given by maximal sub-solution of BSDEs, due to dual-representations \citep{DrapeauTangpi}, are also adequate to guarantee existence of optimal strategies as shown in \citep{tangpi2016}.
    As for the directional differentiability condition at the optimal solution $\hat{\pi}^\ast$, it is necessary to guarantee the identification of the optimal solution with its point-wise version.
    This condition is usually checked on case by case such as for the certainty equivalent.
\end{remark}

\section{Financial Applications and Examples}\label{sec:03}
In the following, we illustrate the characterization of Theorem \ref{thm:maintheorem} to different case study.
We present explicit solutions for the optimal strategy in the complete and incomplete case for a modified exponential utility maximization and an application of which to illustrate the cost of incompleteness in terms of indifference when facing an incomplete market with respect to a complete one.
We conclude by addressing recursive utility optimization which bears some particularity in terms of the gradient conditions.

\subsection{Illustration: Complete versus Incomplete Market}
The running example we will use is inspired from the dual representation in \citep{DrapeauTangpi} where
\begin{equation*}
    U(F)=\inf_{b \in \mathcal{D}, c \in bmo}\left\{ E\left[ M^{bc}_T F+\int_{0}^{T}M^{bc} g^{\ast}(b,c) ds \right] \right\}, \quad F \in L^\infty
\end{equation*}
According to this dual representation in terms of discounting and probability uncertainty, we consider the simple example where
\begin{equation*}
    g^{\ast}(b,c)=
    \begin{cases}
        \frac{\gamma c^2}{2} & \text{if } b = \beta
        \\
        \infty & \text{otherwise}
    \end{cases}
\end{equation*}
where
\begin{itemize}
    \item $\beta$ is a positive bounded predictable process;
    \item $\gamma$ is also a positive predictable process strictly bounded away from $0$ by a constant;
\end{itemize}
Note that even if we consider a discounting factor $\beta$, there is no uncertainty about him particularly.
This is an example of a sub-cash additive valuation instead of the classical cash-additive one, see \citep{ravanelli2009}.
If $\beta=0$ and $\gamma$ is constant, then we have a classical exponential utility optimization problem.
We therefore have
\begin{equation}\label{eq:runningexample}
    g(y,z)=\beta y+\frac{z^2}{2 \gamma}, \quad (y,z)\in \mathbb{R}\times \mathbb{R}^d
\end{equation}
To simplify the comparison between the complete and incomplete market, we assume that we have a simplified market with $d$ stocks following the dynamic
\begin{equation*}
    \frac{dS}{S}=\mu dt + \sigma\cdot dW
\end{equation*}
where $\sigma=Id(d\times d)$ is the identity.
In other terms the randomness driving stock $i$ is the Brownian motion $i$.
It follows that $\theta = \mu$ which is uniformly bounded.
In the complete case, the agent can invest in all the stocks while in the incomplete case it is limited to the first $n$ stocks.
\paragraph{Complete Market:}\label{ex:01}
With the generator $g$ given as in Equation \eqref{eq:runningexample}, it follows that
\begin{equation*}
   \hat \eta(y,z,v)=\gamma\left( v+\theta \right)-z.
\end{equation*}
In particular, $z+\hat\eta(y+x,z,v)=\gamma\left( v+\theta \right)$.
Therefore, in order to find an optimal solution to the optimization problem, since $\partial_y g = \beta$ which is in $\mathcal{D}$, it is sufficient to solve the following coupled forward backward stochastic differential equation
\begin{equation*}
    \begin{cases}
        X_t & =\displaystyle x+\int_0^t\left(\gamma\left(V +\theta\right)-\bar{Z}\right)\cdot \theta ds+\int_0^t\left(\gamma\left(V +\theta\right)-\bar{Z}\right)\cdot dW,  \\
        \bar{Y}_t & =\displaystyle F-\int_t^T\left(\beta(X+\bar{Y})+\frac{\gamma}{2} \left(V^2 -\theta^2\right)+\bar{Z}\cdot\theta\right)ds-\int_t^T\bar{Z}\cdot dW,\\
        U_t & =\displaystyle \int_{0}^{T}\beta ds+\int_t^T\left(\frac{V^2}{2}+V\cdot\theta\right)ds+\int_t^TV \cdot dW,\\
    \end{cases}
\end{equation*}
with solution $X^\ast,\bar{Y}^\ast,\bar{Z}^\ast,U,V$ satisfying
\begin{itemize}
    \item $\pi^\ast=\gamma\left(V +\theta\right)-\bar{Z}^\ast$ is in $bmo$;
    \item $(b^\ast, c^\ast)$ is in $\mathcal{P}^g$ where $b^\ast = \beta$ and $c^\ast = V+\theta$;
    \item $\int M^{bc}\left(\bar{Z}^\ast+\pi^\ast-\left(X^\ast+\bar{Y}^\ast\right)c^\ast\right)\cdot dW$ is in $\mathcal{H}^1$ for all $(b,c)$ in $\mathcal{P}^g$.
\end{itemize}

One can easily deduce that the last backward stochastic differential equation admits a unique solution with $V$ in $bmo$ due to the assumption on $\beta$, see \citep{imkeller2005}.
To provide an explicit solution,
\begin{itemize}
    \item \textbf{We further assume that }$M^{\theta}_T$ \textbf{is bounded}.
\end{itemize}
\begin{remark}
    This is in particular the case if $(Y,\theta)$ is solution of the following quadratic backward stochastic differential equation
    \begin{equation*}
        Y_t=H +\int_t^T\frac{\theta^2}{2}ds+\int_t^T\theta \cdot dW
    \end{equation*}
    for some $H\in L^{\infty}$.
    Indeed, in that case $M^\Theta_T = \exp(H-Y_0)$ which is bounded.
    If in addition $H=f(W_T)$ where $f:\mathbb{R}^d\rightarrow\mathbb{R}$ is a bounded Lipschitz function, then $\theta$ is bounded, see \cite{Cheridito2014}.
    Conversely, since $\theta$ is bounded, hence in $bmo$, if $M^{\theta}_T$ is bounded, $(\ln M^{\theta},\theta)$ is the unique solution of the following backward stochastic differential equation
    \begin{equation*}
        \ln M^{\theta}_t = \ln M^{\theta}_T + \int_t^T\frac{\theta^2}{2}ds+\int_t^{T}\theta \cdot dW.
    \end{equation*}
\end{remark}
Defining
\begin{equation*}
    X_T:=\frac{1}{D_T^\beta}\left(C+\int_{0}^{T}\frac{D^\beta \gamma (V+\theta)^2}{2}dt +\int_0^T D^\beta\gamma (V+\theta)\cdot dW\right)-F
\end{equation*}
and noting that
\begin{align*}
    \int_{0}^{T}\left(\frac{D^\beta V^2}{2}+D^{\beta}\theta V\right)dt +\int_0^T D^\beta V\cdot dW & = U_0 +\int_0^T\beta D^{\beta} Udt-D^{\beta}_T\int_0^T\beta dt\\
    \int_{0}^{T}\frac{D^\beta \theta^2}{2}dt +\int_0^T D^\beta\theta\cdot dW & = \int_0^T\beta D^{\beta}\ln M^\theta dt -D^{\beta}\ln M^{\theta}_T
\end{align*}
it follows that $X_T$ is bounded and we choose the constant $C$ such that $E^\theta[X_T]=x$.\footnote{That is
    \begin{equation*}
        C=\frac{1}{E^\theta\left[ D_T^{-\beta} \right]}\left(x+E^\theta\left[ F \right] -E^\theta\left[ \frac{1}{D_T^\beta}\left(\int_{0}^{T}\frac{D^\beta \gamma (V+\theta)^2}{2}dt +\int_0^T D^\beta \gamma (V+\theta)\cdot dW\right) \right] \right).
    \end{equation*}
}
Thus, by martingale representation theorem, there exists a predictable process $\Gamma$ in $bmo$ such that $X_T=x+\int_0^T\Gamma \cdot dW^{\theta}$.
Defining
\begin{equation*}
    \begin{cases}
        X^\ast & :=\displaystyle x+\int \theta \cdot \Gamma dt+\int \Gamma \cdot dW\\
        \bar{Y}^\ast & :=\displaystyle \frac{1}{D^\beta}\left(C+\int \frac{D^\beta \gamma(V+\theta)^2}{2}ds +\int D^\beta\gamma(V+\theta)\cdot dW\right)-X^\ast\\
        \bar{Z}^\ast & := \displaystyle \gamma\left(V+\theta\right)-\Gamma
    \end{cases}
\end{equation*}
it follows that $(X^\ast, \bar{Y}^\ast, \bar{Z}^\ast, U, V)$ is solution of the forward backward stochastic differential equation.
We are left to check that this solution satisfies the integrability conditions.
First, $\pi^\ast=\Gamma$ is in $bmo$.
Second, $b^\ast=\beta$ is bounded hence $b \in \mathcal{D}$.
Third, $c^*\in bmo$ and
\begin{equation*}
    g^\ast(b^\ast,c^\ast)=\frac{\gamma(c^\ast)^2}{2}=\frac{\gamma(V+\theta)^2}{2}\geq 0
\end{equation*}
therefore, it holds that
\begin{multline*}
    0\leq E\left[\int_0^TM^{b^\ast c^\ast}g^\ast(b^\ast,c^\ast)dt\right]
    =E\left[\int_0^{T}M^{c^\ast}D^{b^\ast}\frac{\gamma(c^\ast)^2}{2}dt\right]\\
    =E\left[\int_0^T-M^{c^\ast}c^\ast \cdot dW\int_0^T-\frac{\gamma D^{b^\ast}c^\ast}{2}\cdot dW\right]
    =E\left[\left(M^{c^\ast}_T-1\right)\int_0^T-\frac{\gamma D^{b^\ast}c^\ast}{2}\cdot dW\right]<+\infty.
\end{multline*}
Thus $(b^\ast,c^\ast)\in\mathcal{P}^g$.
Finally, in order to show that $\int M^{bc}\left(\bar{Z}^\ast+\pi^\ast-\left(X^\ast+\bar{Y}^\ast\right)c^\ast\right)\cdot dW$ is in $\mathcal{H}^1$ for all $(b,c)$ in $\mathcal{P}^g$, according to Remark \ref{remark_admin},  we only need to check that for every $(b,c)\in\mathcal{P}^g$, $\sup_{0\leq t\leq T}|M^{bc}_t(X^{\ast}_t+\bar{Y}^{\ast}_t)|$ is in $L^1$, which follows directly from a similar technique as in Lemma \ref{lem:integrabilityterminal} noting that
\begin{equation*}
X^\ast + \bar{Y}^\ast=\displaystyle \frac{1}{D^\beta}\left(C+\int \frac{D^\beta \gamma(V+\theta)^2}{2}ds +\int D^\beta\gamma(V+\theta)\cdot dW\right).
\end{equation*}
Thus, $\pi^\ast=\Gamma=\gamma(V+\theta)-Z^\ast$ is an optimal solution to the optimization problem.
\begin{remark}
    In terms of utility optimization, since $U(F+X_T^{\pi})=\bar{Y}^\ast_0+x$, it follows that
    \begin{multline}\label{eq:result1}
        U\left( F+X_{T}^{\pi^\ast}  \right)=\bar{Y}_0^\ast+x=C\\
        =\frac{1}{E^\theta\left[ D_T^{-\beta} \right]}\left(x+E^\theta\left[ F \right] -E^\theta\left[ \frac{1}{D_T^\beta}\left(\int_{0}^{T}\frac{D^\beta \gamma (V+\theta)^2}{2}dt +\int_0^T D^\beta \gamma (V+\theta)\cdot dW\right) \right] \right)
    \end{multline}
\end{remark}
\begin{remark}
    Instead of assuming that $M^{\theta}_T$ is bounded, we can still have an explicit solution if $\beta$ is deterministic similarly as in the incomplete market, in which we will give the detailed method to get the solution.
\end{remark}
\paragraph{Incomplete Market:}\label{ex:02}
Still with the generator $g$ given as in Equation \eqref{eq:runningexample} but now in the incomplete case -- that is $n<d$ and $\hat{\theta}=\hat{\mu}$ -- it follows that
\begin{equation*}
    \hat{\eta}(y,z,\hat{v})=\gamma \left( \hat{v}+\hat{\theta} \right)-\hat{z}.
\end{equation*}
In particular, $\hat{z}+\hat{\eta}(y+x,z,v)=\gamma\left( \hat{v}+\hat{\theta} \right)$.
Here again $\partial_{y}g = \beta$, and since $\partial_{\tilde{z}}g = \tilde{z}/{\gamma}$, in order to find an optimal solution to the optimization problem, it is sufficient to solve the following coupled forward backward stochastic differential equation
\begin{equation*}
    \begin{cases}
        X_t & = \displaystyle x+\int_0^t\hat{\theta}\cdot \left(\gamma\left(\hat{V} +\hat{\theta}\right)-\hat{\bar{Z}}\right) ds+\int_0^t\left(\gamma\left(\hat{V} +\hat{\theta}\right)-\hat{\bar{Z}}\right)\cdot d\hat{W},  \\
        \bar{Y}_t & = \displaystyle F-\int_t^T\left(\beta(X+\bar{Y})+\frac{\gamma}{2} \left(\hat{V}^2 -\hat{\theta}^2\right)+\hat{\theta}\cdot \hat{\bar{Z}}+\frac{\tilde{\bar{Z}}^2}{2\gamma}\right)ds-\int_t^T \bar{Z} \cdot dW,\\
        U_t & =\displaystyle U_T+\int_t^T\left(\frac{\hat{V}^2-\tilde{V}^2}{2}+\hat{\theta}\cdot \hat{V}\right)ds+\int_t^TV \cdot dW,\\
        U_T & = \displaystyle \int_{0}^{T}\left(\beta+\frac{\tilde{\bar{Z}}^2}{2\gamma^2}\right) ds+\int_{0}^{T}\frac{\tilde{\bar{Z}}}{\gamma} \cdot d\tilde{W}
    \end{cases}
\end{equation*}
with solution $X^\ast,\bar{Y}^\ast,\bar{Z}^\ast,U,V$ satisfying
\begin{itemize}
    \item $\hat{\pi}^\ast=\gamma\left(\hat{V} +\hat{\theta}\right)-\hat{\bar{Z}}^\ast$ is in $bmo$;
\item $(b^\ast,c^\ast)\in\mathcal{P}^g$ where $b^\ast = \beta$ and $c^\ast=(\hat{V}+\theta,\tilde{\bar{Z}}^\ast/\gamma)$;
    \item $\int M^{bc}(\bar{Z}^\ast + \pi^\ast-(X^\ast+\bar{Y}^\ast)c^\ast)\cdot dW$ is in $\mathcal{H}^1$ for all $(b,c)\in\mathcal{P}^g$.
\end{itemize}
In order to provide an explicit solution as in the complete market
\begin{itemize}
    \item \textbf{we assume here that} $\beta$ \textbf{is deterministic};
\end{itemize}
First, if we assume a-priori that $\tilde{c}^\ast=\tilde{\bar{Z}}^\ast/\gamma$ is in $\tilde{bmo}$, since $\beta$ is deterministic, the last backward stochastic differential equation admits a unique solution with $V=(0,-\tilde{c})$ in $bmo$.
Indeed, the following quadratic BSDE
    \begin{equation*}
        \begin{cases}
            \bar{\Upsilon}_t & = \displaystyle \bar{\Upsilon}_T-\int_{t}^{T}\left(\hat{\theta}\cdot\hat{\Lambda}+\frac{\tilde{\Lambda}^2}{D^\beta 2 \gamma} \right) ds-\int_{t}^{T} \Lambda \cdot dW  \\
            \bar{\Upsilon}_T & = \displaystyle D_{T}^\beta \left(F+x\right)+\int_{0}^{T}\frac{D^\beta \gamma \hat{\theta}^2}{2}dt
        \end{cases}
    \end{equation*}
    admits a unique solution with $\Lambda$ in $bmo$ since $\bar{\Upsilon}_T$ is bounded, see \citep{imkeller2005}.
    Therefore,
    \begin{equation*}
        \begin{cases}
            \Upsilon_t &= \displaystyle \bar{\Upsilon}_t -\int_{0}^{t}\gamma\hat{\theta}^2\left(D^\beta-D^\beta_T\right)ds- \int_{0}^{t}\gamma \left(D^\beta-D^\beta_T\right) \hat{\theta}\cdot d\hat{W}\\
            \hat{\Gamma}_t &= \hat{\Lambda}_t - \gamma\left(D^\beta_t-D^\beta_T\right)\hat{\theta}_t\\
            \tilde{\Gamma}_t &= \tilde{\Lambda}_t
        \end{cases}
    \end{equation*}
    satisfies the following quadratic BSDE
    \begin{equation*}
        \begin{cases}
            \Upsilon_t & = \displaystyle \Upsilon_T-\int_{t}^{T}\left(\hat{\theta}\cdot \hat{\Gamma}+\frac{\tilde{\Gamma}^2}{D^\beta 2 \gamma}\right) ds-\int_{t}^{T} \Gamma \cdot dW  \\
            \Upsilon_T & = \displaystyle D_{T}^\beta \left(F+x\right)+\int_{0}^{T}\frac{D^\beta \gamma \hat{\theta}^2}{2}dt - \int_{0}^{T} \gamma \left(D^\beta-D^\beta_T\right)\hat{\theta}^2dt - \int_{0}^{T}\gamma\left(D^\beta-D^\beta_T\right) \hat{\theta}\cdot d\hat{W}
        \end{cases}
    \end{equation*}
    with $\Gamma$ in $bmo$.

It follows that the system is solved for
\begin{equation*}
    \begin{cases}
        \hat{\pi}^\ast & = \displaystyle\gamma \hat{\theta} - \hat{\bar{Z}}^\ast=\gamma\hat{\theta}-\frac{\hat{\Gamma}}{D_T^\beta}\\
        \tilde{c} & = \displaystyle \frac{\tilde{Z}^\ast}{\gamma}=-\tilde{V}\\
        \hat{\bar{Z}}^\ast &= \displaystyle \gamma \hat{\theta} -\hat{\pi}^\ast=\frac{\hat{\Gamma}}{D_T^\beta} \\
        \tilde{\bar{Z}}^\ast &= \displaystyle\frac{\tilde{\Gamma}}{D^\beta}\\
        \hat{V} & = \displaystyle 0\\
        X^\ast &= \displaystyle x+ \int_{}^{} \hat{\pi}^\ast\cdot \hat{\theta} dt+ \int_{}^{} \hat{\pi}^\ast \cdot d\hat{W}\\
        \bar{Y}^\ast & = \displaystyle \frac{1}{D^\beta}\left( \Upsilon_0+\int \frac{D^\beta \gamma \hat{\theta}^2}{2}ds+\int \frac{\tilde{\Gamma}^2}{D^\beta 2 \gamma} ds+\int D^\beta \gamma\hat{\theta} \cdot d\hat{W} +\int \tilde{\Gamma} \cdot d\tilde{W}   \right)-X^\ast
    \end{cases}
\end{equation*}
The fact that the conditions of Theorem \ref{thm:maintheorem} are fulfilled follows the same argumentation as in complete market.
\begin{remark}
    Again, in terms of utility optimization, we obtain that
    \begin{equation}\label{eq:result2}
        U\left( F+X_T^{\hat{\pi}^\ast} \right) = \bar{Y}_0^\ast+x=D_T^\beta x+E^{\hat{\theta}}\left[ D_T^\beta F+\int_{0}^{T}\frac{D^\beta\gamma\hat{\theta}^2}{2}dt-\int_{0}^{T}\frac{\tilde{\Gamma}^2}{D^\beta 2\gamma}dt   \right]
    \end{equation}
\end{remark}

\paragraph{The Cost of Incompleteness}
The computation of explicit portfolio optimal strategies allows to address further classical financial problems such as utility indifference pricing.
Given a contingent claim $F$, we are looking at the start wealth $x^\ast$ such that
\begin{equation*}
    U(F)=U\left( F+x^\ast+\int_{0}^{T}\hat{\pi}^\ast \cdot d\hat{W}^{\hat{\theta}}  \right)
\end{equation*}
where $\hat{\pi}^\ast$ is the corresponding optimal strategy.
In other terms, $x^\ast$ represents the value in terms of indifference pricing one is willing to pay to reach the same utility by having access to a financial market.
Since our functional is only upper semi-continuous, and to distinguish between complete and incomplete markets we proceed as follows.
\begin{equation*}
    \begin{split}
        x^\ast & =\inf\left\{ x \in \mathbb{R}\colon \sup_{\pi \in bmo}U\left( F+x+\int_{0}^{T}\pi\cdot dW^\theta  \right)> U(F) \right\}\\
               & =\inf\left\{ x \in \mathbb{R}\colon U\left( F+x+\int_{0}^{T}\pi\cdot dW^\theta  \right)> U(F) \quad\text{for some }\pi \in bmo\right\}\\
        y^\ast & =\inf\left\{ y \in \mathbb{R}\colon \sup_{\hat{\pi}\in \hat{bmo}}U\left(F+y+\int_{0}^{T}\hat{\pi} \cdot d\hat{W}^{\theta}\right)> U(F) \right\}\\
               & =\inf\left\{ y \in \mathbb{R}\colon U\left(F+y+\int_{0}^{T}\hat{\pi} \cdot d\hat{W}^{\theta}\right)> U(F) \quad \text{for some }\hat{\pi} \in \hat{bmo}\right\}
    \end{split}
\end{equation*}
which represents the utility indifference amount of wealth to be indifferent for $F$ in the complete and incomplete case, respectively.
Intuitively, the amount of wealth necessary to reach the same utility level is higher in the incomplete case, that is $x^\ast \leq y^\ast$.
This is indeed the case since $\hat{bmo}$ is a subset of $bmo$.

In the case of the previous example where an explicit solution stays at hand we have the following explicit costs of having a restricted access to the financial market.
Indeed, in the case where $\beta$ is deterministic, according to Equations \ref{eq:result1} and \ref{eq:result2} we obtain
\begin{equation*}
    U\left( F+x^\ast+\int_{0}^{T}\pi^\ast \cdot dW^\theta  \right)=D_T^{\beta}x^\ast+ E^\theta\left[ D_T^\beta F + \int_{0}^{T}\frac{D^\beta \gamma\theta^2}{2}dt  \right].
\end{equation*}
On the other hand, according to \eqref{eq:result2} it holds
\begin{equation*}
    U\left( F+y^\ast+\int_{0}^{T}\hat{\pi}^\ast \cdot d\hat{W}^{\hat{\theta}}  \right)=D_T^\beta y^\ast+E^{\hat{\theta}}\left[ D_T^\beta F+\int_{0}^{T}\frac{D^\beta\gamma\hat{\theta}^2}{2}dt-\int_{0}^{T}\frac{\tilde{\Gamma}^2}{D^\beta 2\gamma}dt   \right].
\end{equation*}
We deduce that
\begin{equation*}
    \begin{cases}
        x^\ast & = \displaystyle \frac{U(F)}{D_T^\beta}-\frac{1}{D_T^{\beta}}E^{\theta}\left[ D_T^\beta F+\int_{0}^{T}\frac{D^\beta\gamma\theta^2}{2}dt  \right]\\
        \\
        y^\ast & = \displaystyle \frac{U(F)}{D_T^\beta}-\frac{1}{D_T^\beta}E^{\hat{\theta}}\left[ D_T^\beta F+\int_{0}^{T}\frac{D^\beta\gamma\hat{\theta}^2}{2}dt   \right]+\frac{1}{D_T^\beta}E^{\hat{\theta}}\left[ \int_{0}^{T}\frac{\tilde{\Gamma}^2}{ D^\beta 2\gamma}dt   \right]
    \end{cases}
\end{equation*}

\subsection{Inter-Temporal Resolution of Uncertainty}
We conclude with a classical utility functional having some interesting particularity in terms of gradient characterization.
To address inter-temporal resolution of uncertainty, \citet{kreps1978} introduced a new class of inter-temporal utilities that weight immediate consumption against later consumptions and random payoffs.
This idea has been extended in particular by \citet{epstein1989} in the discrete case and later on by \citet{duffie1992} in the continuous case in terms of backward stochastic differential equations.
Given a cumulative consumption stream $c$, positive increasing and right continuous function, a commonly used example of inter-temporal generator of a recursive utility is given by
\begin{equation*}
    f(c,y)=\frac{\beta}{\rho}\frac{c^{\rho}-(\alpha y)^{\rho/\alpha}}{(\alpha y)^{\rho/\alpha-1}}
\end{equation*}
where $\rho, \alpha \in (0,1)$ and $\beta \geq 0$.
We refer to \citep{duffie1992} for the interpretation, properties and derivation of this generator and the corresponding constants.
Note that this generator is concave in $y$ if $\rho\leq \alpha \leq 1$, assumption we will keep.
In the classical setting, the generator is represented in terms of utility with a positive sign in the backward stochastic differential equation.
In our context in terms of costs with $0<\rho\leq \alpha\leq 1$, and $\beta \geq 0$ we define
\begin{equation*}
    g(y)=
    \begin{cases}
        \displaystyle \frac{\beta}{\rho}\frac{(\alpha y)^{\rho/\alpha}-c^{\rho}}{(\alpha y)^{\rho/\alpha-1}}=\frac{\beta \alpha}{\rho} \left(y-\gamma y^{1-\rho/\alpha}\right)& \text{if } y\geq 0\\
        \infty & \text{otherwise}
    \end{cases}
\end{equation*}
which is a convex function in $y$ and where $\gamma = c^\rho/\alpha^{\rho/\alpha}$.
In terms of costs, given a deterministic right continuous increasing consumption stream $c$, the agent weight infinitesimally the opportunity to consume today weighted with a parameter $\rho$ with a rest certainty equivalent of consumption tomorrow to the power $\rho/\alpha$ against the cost in terms of certainty equivalent if waiting tomorrow and not consuming.
The recursive utility $U(F)$ with terminal payoff $F$ is given as the maximal sub-solution of
\begin{equation*}
    \begin{cases}
        Y_s &\leq \displaystyle  Y_t-\int_{s}^{t}g(Y)ds-\int_{s}^{t}Z \cdot dW\\
        Y_T &= F
    \end{cases}
\end{equation*}
In this context, given a random payoff $F$, start wealth $x$, and consumption stream $c$, the agent tries to optimize its recursive utility $U(F+X_T^{\hat{\pi}})$ in terms of investment strategy $\hat{\pi}$ against its consumption choice $c$.
For the sake of simplicity we consider the simple case of a complete market.
The particularity of recursive utilities is that the generator usually do not depend on $z$.
It follows that the condition $\partial_{z} g=0=v+\theta$ enforces the condition in terms of auxiliary backward stochastic differential equation
\begin{equation*}
    U_t=\int_{0}^{T}\partial_{y}g(X+\bar{Y})ds-\int_{t}^{T}\theta^2ds-\int_{t}^{T}\theta\cdot dW.
\end{equation*}
Since
\begin{equation*}
    \partial_yg(X+\bar{Y})=\frac{\beta\alpha}{\rho}\left( 1-\gamma\left(1-\frac{\rho}{\alpha} \right)(X+\bar{Y})^{-\rho/\alpha}\right)
\end{equation*}
we can assume that $X+\bar{Y} = \Phi$ where $t\mapsto \Phi(t)$ is a deterministic function.
Then it follows that
\begin{equation*}
    \bar{Y}_t = \displaystyle F-\int_{t}^{T}\left(\frac{\beta \alpha}{\rho} \left( \Phi-\gamma \Phi^{1-\rho/\alpha} \right)-\pi^\ast\cdot\theta\right)ds -\int_{t}^{T}\bar{Z}\cdot  dW
\end{equation*}
showing that
\begin{equation*}
    F = \bar{Y}_0+\int_{0}^{T}\left(\frac{\beta \alpha}{\rho} \left( \Phi-\gamma \Phi^{1-\rho/\alpha} \right)-\pi^\ast\cdot\theta\right) dt +\int_{0}^{T}\bar{Z} \cdot dW.
\end{equation*}
Setting
\begin{equation*}
    \begin{cases}
        \pi^\ast & = -\bar{Z}\\
        X^\ast & = \displaystyle x+\int_{}^{}\pi^*\cdot\theta dt+\int_{}^{} \pi^\ast \cdot dW\\
        \bar{Y}^\ast & =  \displaystyle \bar{Y}_0 +\int \left(\frac{\beta \alpha}{\rho} \left( \Phi-\gamma \Phi^{1-\rho/\alpha} \right)-\pi^\ast\cdot\theta\right)dt +\int \bar{Z}\cdot  dW\\
        \bar{Z}^\ast & = \bar{Z}
    \end{cases}
\end{equation*}
from $X+\bar{Y} = \Phi$, we deduce that if $\Phi$ is solution of the ordinary differential equation
\begin{equation*}
    \begin{cases}
        \Phi^\prime & = \displaystyle \frac{\beta\alpha}{\rho}\left( \Phi-\gamma \Phi^{1-\rho/\alpha} \right)\\
        \Phi(T) & = E^{\theta}[F]+x
    \end{cases}
\end{equation*}
then the system has an optimal solution.

\begin{appendix}
	\section{Existence and uniqueness of maximal sub-solutions}\label{sec:app01}
\begin{proof}[of Theorem \ref{thm:maxsubsol}]
    Throughout this proof, we use the notation
    \begin{equation*}
        \mathcal{A}^0 = \left\{ (Y,Z)\in \mathcal{S}\times \mathcal{L}\colon
            \begin{cases}
                (Y,Z) \text{ satisfies } \eqref{eq:03}\\
                \displaystyle \int_{}^{} M^{bc}(Z -Yc)\cdot dW\text{ is a sub-martingale for every }(b,c) \in \mathcal{P}^g
            \end{cases}
        \right\}
    \end{equation*}
    Recall that $(Y,Z)$ is a sub-solution if and only if there exists an adapted c\`adl\`ag increasing process $K$ with $K_0=0$ such that
    \begin{equation*}
        Y_t=F-\int_{t}^{T}g(Y,Z)ds-(K_T-K_t)-\int_{t}^{T}Z\cdot dW
    \end{equation*}
    which is given by
    \begin{equation}\label{eq:10}
        K_t=Y_t-Y_0-\int_{0}^{t} g(Y,Z)ds-\int_{0}^{t}Z\cdot dW.
    \end{equation}
    We prove the theorem in several steps
    \begin{enumerate}[label=\textit{Step \arabic*:}, fullwidth]
        \item For any $(Y,Z)$ in $\mathcal{A}(F)$ and $(b,c)$ in $\mathcal{P}^g$, defining $\check{Y} = M^{bc}Y+\int_{}^{} M^{bc}g^\ast(b,c)dt$, and $\check{Z} = M^{bc}(Z-Yc)$, it follows that $\sup_{0\leq t\leq T}| \check{Y}_t | \in L^1$.
            Indeed, using Ito's formula, it follows that $(\check{Y}, \check{Z})$ satisfies
            \begin{equation*}
                \begin{cases}
                    \check{Y}_s & \leq \displaystyle \check{Y}_t - \int_{s}^{t}\check{g}(\check{Y}, \check{Z}) du -\int_{s}^{t}\check{Z} \cdot dW  \\
                    \check{Y}_T & = \displaystyle M^{bc}_T F +\int_{0}^{T}M^{bc} g^\ast(b,c)dt
                \end{cases}
            \end{equation*}
            where
            \begin{multline}\label{eq:modifieddriver}
                \check{g}(\check{y}, \check{z}) := M^{bc}\left[g\left(\frac{\check{y} - \int_{}^{}M^{bc}g^\ast(b,c)dt}{M^{bc}}, \frac{\check{z} + \left(\check{y} - \int_{}^{} M^{bc}g^\ast(b,c)dt\right)c}{M^{bc}}\right)\right.\\
                \left. -b\frac{\check{y} - \int_{}^{}M^{bc}g^\ast(b,c)dt}{M^{bc}} -c\cdot \frac{\check{z} + \left(\check{y} - \int_{}^{} M^{bc}g^\ast(b,c)dt\right)c}{M^{bc}}+g^\ast(b,c) \right]\geq 0.
            \end{multline}
            On the one hand, since $\check{g}$ is positive and $\int_{}^{} \check{Z} \cdot dW$ is in $\mathcal{H}^1$, it holds
            \begin{equation*}
                \check{Y}_t \geq \check{Y}_0 + \int_{0}^{t} \check{g}(\check{Y}, \check{Z}) du +\int_{0}^{t}\check{Z}\cdot  dW
                \geq \check{Y}_0  -\sup_{0\leq t\leq T}\left|  \int_{0}^{t}\check{Z}\cdot  dW\right| \in L^1
            \end{equation*}
            On the other hand, once again since $\check{g}$ is positive, by assumption $\int_{}^{} \check{Z}\cdot dW$, $E[M^{bc} F |\mathcal{F}_\cdot]$ as well as $E[ \int_{0}^{T}M^{bc}g^\ast(b,c)dt | \mathcal{F}_{\cdot} ]$ are in $\mathcal{H}^1$, we have
            \begin{align*}
                \check{Y}_t & \leq E\left[ M^{bc}_T F +\int_{0}^{T}M^{bc} g^\ast(b,c)dt -\int_{t}^{T}\check{g}(\check{Y}, \check{Z}) du -\int_{t}^{T}\check{Z} \cdot dW\big | \mathcal{F}_t\right] \\
                          & \leq \sup_{0\leq t \leq T}\left|E\left[ M^{bc}  F \big| \mathcal{F}_t\right]\right|+\sup_{0\leq t\leq T}\left| E\left[\int_{0}^{T}M^{bc}g^\ast(b,c)dt  \big | \mathcal{F}_t\right]\right| \in L^1
            \end{align*}
            showing that $\sup_{0\leq t\leq T}|\check{Y}_t|$ is in $L^1$.
        \item Let $(Y^n,Z^n)$ be a sequence in $\mathcal{A}(F)$, $\tau$ a stopping time and $(B^n)$ be a partition in $\mathcal{F}_{\tau}$.
            Suppose that $Y^1_{\tau}1_{B^n}\leq Y^n_\tau1_{B^n}$, for all $n=1, 2, \ldots$, then it follows that
            \begin{align*}
                Y^0 & := Y^1 1_{[0, \tau)}+\sum Y^n 1_{[\tau, T]} 1_{B^n}\\
                Z^0 & := Z^1 1_{[0, \tau]}+\sum Z^n 1_{(\tau, T]} 1_{B^n}
            \end{align*}
            is such that $(Y^0, Z^0)$ is in $\mathcal{A}^0$.
            It is clear that $(Y^0,Z^0)$ satisfies the sub-solution system \eqref{eq:03}.
            Let us show that it is admissible in the sense of $\mathcal{A}^0$.
            For $(b,c)$ in $\mathcal{P}^g$ we denote $\check{Y}^n = M^{bc}Y^n+\int_{}^{} M^{bc}g^\ast(b,c)dt$, and $\check{Z}^n = M^{bc}(Z^n-Y^nc)$ for every $n=0, 1, \ldots$.
            From the previous computations, we have that $\check{Y}^n_t \leq H$ for every $n=1, 2, \ldots$ where
            \begin{equation*}
                H =\sup_{0\leq t \leq T} \left| E\left[ M^{bc} F \big| \mathcal{F}_t\right]\right|+\sup_{0\leq t\leq T}\left| E\left[\int_{0}^{T}M^{bc}g^\ast(b,c)dt \big | \mathcal{F}_t\right]\right| \in L^1
            \end{equation*}
            We deduce from \eqref{eq:10} that
            \begin{equation*}
                \int_{0}^{t} \check{Z}^0 \cdot dW \leq \check{Y}^0_t-\check{Y}^0_0 \leq H -Y^1_0 \in L^1
            \end{equation*}
            and therefore $\int_{}^{} \check{Z}^0 dW$ is a sub-martingale.
        \item The same argumentation can be done when the sequence $(Y^n,Z^n)$ is taken in $\mathcal{A}^0$ since we only look at the sub-martingale property.
            Hence, $\mathcal{A}^0$ is stable under upward pasting.
        \item Following the construction in \citep[Step 2 and 8 of Theorem 4.1]{heyne2013}, we get a sequence $(Y^n,Z^n)$ of elements in $\mathcal{A}^0$ such that $Y^n \uparrow Y$ the \cadlag version of the time wise essential supremum along dyadic partition of $\mathcal{A}^0$.
            Furthermore, fixing $(b_0,c_0)$ in $\mathcal{P}^g$, we can follow \citep[Step 3 - 7 of Theorem 4.1]{heyne2013} to construct a subsequence in the asymptotic convex hull such that $Z^n \to Z$ $P\otimes dt$-almost surely.
            Following \citep[Step 8 and first part of Step 9 of Theorem 4.1]{heyne2013}, we can verify that $(Y,Z)$ is a sub-solution for the system \eqref{eq:03}.
            We are left to verify the admissibility condition for $Z$.
        \item Note that in the construction of the approximating sequence $(Y^n,Z^n)$ where $Y^n\uparrow Y$, since $\mathcal{A}(F)$ is non empty, we can assume that $(Y^1,Z^1)$ is in $\mathcal{A}(F)$ and it holds $Y^1\leq Y$.
            Let $(b,c)$ in $\mathcal{P}^g$ and denote by $\check{Y} = M^{bc}Y+\int_{}^{} M^{bc}g^\ast(b,c)dt$, and $\check{Z} = M^{bc}(Z-Yc)$ as well as $\check{Y}^n = M^{bc}Y^n+\int_{}^{} M^{bc}g^\ast(b,c)dt$, and $\check{Z}^n = M^{bc}(Z^n-Y^nc)$ for all $n=1, 2, \ldots$.
            On the one hand, since $\check{g}$ as defined in \eqref{eq:modifieddriver} in the first step is positive, $\int_{}^{} \check{Z}^n dW$ is a sub-martingale, $E[M^{bc} F|\mathcal{F}_\cdot]$ as well as $E[ \int_{0}^{T}M^{bc}g^\ast(b,c)dt | \mathcal{F}_{\cdot} ]$ are in $\mathcal{H}^1$, we have
            \begin{align*}
                \check{Y}^n_t & \leq E\left[ M^{bc}_T F+\int_{0}^{T}M^{bc} g^\ast(b,c)dt -\int_{t}^{T}\check{g}(\check{Y}^n, \check{Z}^n) du -\int_{t}^{T}\check{Z}^n \cdot dW\big | \mathcal{F}_t\right] \\
                          & \leq \sup_{0\leq t \leq T}\left|E\left[ M^{bc} F \big| \mathcal{F}_t\right]\right|+\sup_{0\leq t\leq T}\left|E\left[ \int_{0}^{T}M^{bc}g^\ast(b,c)dt \big | \mathcal{F}_t\right]\right| \in L^1
            \end{align*}
            Therefore, according to \eqref{eq:10}, it holds
            \begin{multline*}
                \int_{0}^{t} \check{Z} \cdot dW  \leq \check{Y}_t - \check{Y}_0
                                         = \sup_{n} \check{Y}^n_t - \check{Y}_0\\
                                         \leq \sup_{0\leq t \leq T}\left|E\left[ M^{bc} F \big| \mathcal{F}_t\right]\right|+\sup_{0\leq t\leq T}\left|E\left[ \int_{0}^{T}M^{bc}g^\ast(b,c)dt \big | \mathcal{F}_t\right]\right|- \check{Y}_0 =: H_1\in L^1
            \end{multline*}
            In particular $\int_{}^{} \check{Z}\cdot dW$ is a sub-martingale and therefore $(Y, Z)$ is in $\mathcal{A}^0$.
            On the other hand, following \eqref{eq:10}, defining
            \begin{equation*}
                \check{A}_t=K_t+\int_{0}^{t}\check{g}(\check{Y}, \check{Z}) ds = \check{Y}_t-\check{Y}_0-\int_{0}^{t}\check{Z} \cdot dW
            \end{equation*}
            It follows that $\check{A}$ is an increasing \cadlag process starting at $0$.
            Using the sub-martingale property of $\int_{}^{} \check{Z}\cdot dW$, we get that
            \begin{equation*}
                E[\check{A}_T ] \leq E\left[ \check{Y}_T \right]-\check{Y}_0 = E\left[ M^{bc}_T F+\int_{0}^{T}M^{bc} g^\ast(b,c)dt \right]-\check{Y}_0<\infty
            \end{equation*}
            Hence $\check{A}_T$ is in $L^1$.
            It follows that
            \begin{equation*}
                \int_{0}^{t}\check{Z} \cdot dW \geq \check{Y}_t -\check{Y}_0 - \check{A}_T \geq -\sup_{0\leq t\leq T}\left| \check{Y}^1_t \right| -\check{Y}_0-\check{A}_T =:H_2
            \end{equation*}
            Since $(Y^1, Z^1)$ is in $\mathcal{A}$, according to the first step we have that $\sup_{0\leq t\leq T}|\check{Y}_t^1|$ is in $L^1$.
            Hence $H_2$ is in $L^1$.
            From both inequality we deduce that
            \begin{equation*}
                \sup_{0\leq t\leq T}\left| \int_{0}^{t}\check{Z} \cdot dW  \right|\leq H_1- H_2 \in L^1
            \end{equation*}
    \end{enumerate}
\end{proof}
\begin{remark}\label{remark_admin}
    Note that from this proof, if $(Y,Z)$ is a sub-solution \eqref{eq:03}, then for every $(b,c)$ in $\mathcal{P}^g$, $\int_{}^{} M^{bc}( Z-Yc)\cdot dW$ is in $\mathcal{H}^1$ if and only if $\sup_{0\leq t\leq T}|M^{bc}_tY_t+\int_0^tM^{bc}g^\ast(b,c)ds|$ is in $L^1$.
    Moreover, if $g^\ast(b,c)\geq 0$ --- which is the case whenever $g(0,0)=0$ --- since $E[\int_0^TM^{bc}g^\ast(b,c)ds|\mathcal{F}_{\cdot}]$ is in $\mathcal{H}^1$ and
    \begin{multline*}
        \sup_{0\leq t\leq T}\left|M^{bc}_t Y_t+\int_0^tM^{bc}g^\ast(b,c)ds\right|\leq \sup_{0\leq t\leq T}\left|M^{bc}_tY_t\right|+\sup_{0\leq t\leq T}\left|\int_0^tM^{bc}g^\ast(b,c)ds\right|\\
        =\sup_{0\leq t\leq T}\left|M^{bc}_tY_t\right|+\int_0^TM^{bc}g^\ast(b,c)ds,
    \end{multline*}
    it follows that $\sup_{0\leq t\leq T}|M^{bc}_tY_t|$ is in $L^1$ implies $\sup_{0\leq t\leq T}|M^{bc}_tY_t+\int_0^tM^{bc}g^\ast(b,c)ds|$ is in $L^1$.
\end{remark}

\end{appendix}

\bibliographystyle{abbrvnat}
\bibliography{biblio}

\end{document}